\makeatletter
\let\@twosidetrue\@twosidefalse
\let\@mparswitchtrue\@mparswitchfalse
\makeatother

\documentclass{llncs}
\usepackage[left=3.5cm,right=2.5cm,top=2.5cm,bottom=2.5cm]{geometry}
\usepackage[latin2]{inputenc}
\usepackage[english]{babel}
\usepackage{dsfont, complexity}
\usepackage{amsmath, lmodern}
\usepackage[title]{appendix}
\usepackage{amssymb}
\usepackage{lmodern}
\usepackage{float}
\usepackage{graphicx}
\usepackage{tikz}
\usepackage{tkz-graph}
\usepackage{xcolor}
\usepackage{url}
\usepackage{paralist, enumerate}
\usepackage{refcount}
\usepackage{longtable}

\usepackage{amsfonts}
\usepackage{multirow}
\usepackage{cancel}
\usepackage{soul} 
\usepackage[all]{xy}
\usepackage{subfigure}
\usepackage{color}
\usepackage{listings}
\usepackage{notoccite}
\usepackage{mathtools} 
\usepackage[normalem]{ulem}  
\usepackage{caption}
\usepackage{multicol}

\usepackage{thmtools, mathtools}
\usepackage{thm-restate}

\usepackage{algorithm}
\DeclareMathOperator{\rank}{rank}
\usepackage[textsize=footnotesize]{todonotes}

\usepackage[noend]{algpseudocode}
\usepackage{etoolbox}

\makeatletter
\newcommand*{\algrule}[1][\algorithmicindent]{%
  \makebox[#1][l]{%
    \hspace*{.2em}
    \vrule height .75\baselineskip depth .25\baselineskip
  }
}

\newcount\ALG@printindent@tempcnta
\def\ALG@printindent{%
    \ifnum \theALG@nested>0
    \ifx\ALG@text\ALG@x@notext
    \else
    \unskip
    \ALG@printindent@tempcnta=1
    \loop
    \algrule[\csname ALG@ind@\the\ALG@printindent@tempcnta\endcsname]%
    \advance \ALG@printindent@tempcnta 1
    \ifnum \ALG@printindent@tempcnta<\numexpr\theALG@nested+1\relax
    \repeat
    \fi
    \fi
}
\patchcmd{\ALG@doentity}{\noindent\hskip\ALG@tlm}{\ALG@printindent}{}{\errmessage{failed to patch}}
\patchcmd{\ALG@doentity}{\item[]\nointerlineskip}{}{}{} 
\makeatother

\algnewcommand\algorithmicforeach{\textbf{for each}}
\algdef{S}[FOR]{ForEach}[1]{\algorithmicforeach\ #1\ \algorithmicdo}

\title{Super-stability in the \\ Student-Project Allocation Problem with Ties\thanks{A preliminary version of a part of this paper appeared in \cite{OM18}.} \thanks{S.~Olaosebikan: Supported by a College of Science and Engineering Scholarship, University of Glasgow. Orcid ID: 0000-0002-8003-7887. D.~Manlove: Supported by grant EP/P028306/1 from the Engineering and Physical Sciences Research Council. Orcid ID: 0000-0001-6754-7308.}}

\author{
Sofiat Olaosebikan \and David Manlove
}
\institute{School of Computing Science, University of Glasgow, \\e-mail: \texttt{Sofiat.Olaosebikan@glasgow.ac.uk, David.Manlove@glasgow.ac.uk}
}

\begin{document}

\maketitle              

\begin{abstract}
The \emph{Student-Project Allocation problem with lecturer preferences over Students} ({\sc spa-s}) involves assigning students to projects based on student preferences over projects, lecturer preferences over students, and the maximum number of students that each project and lecturer can accommodate. This classical model assumes that each project is offered by one lecturer and that preference lists are strictly ordered. Here, we study a generalisation of {\sc spa-s} where ties are allowed in the preference lists of students and lecturers, which we refer to as the \emph{Student-Project Allocation problem with lecturer preferences over Students with Ties} ({\sc spa-st}). We investigate stable matchings under the most robust definition of stability in this context, namely \emph{super-stability}. We describe the first polynomial-time algorithm to find a super-stable matching or to report that no such matching exists, given an instance of {\sc spa-st}. Our algorithm runs in $O(L)$ time, where $L$ is the total length of all the preference lists. Finally, we present results obtained from an empirical evaluation of the linear-time algorithm based on randomly-generated {\sc spa-st} instances. Our main finding is that, whilst super-stable matchings can be elusive when ties are present in the students' and lecturers' preference lists, the probability of such a matching existing is significantly higher if ties are restricted to the lecturers' preference lists.


\end{abstract}

\thispagestyle{empty}
\setcounter{page}{1}
\pagestyle{headings}
\section{Introduction}
\label{introduction}
The \emph{Student-Project Allocation problem} ({\sc spa}) \cite{AIM07,CFG19,Man13} is a many-one matching problem which involves three sets of entities: students, projects and lecturers. Each project is proposed by one lecturer and each student is required to rank a subset of these projects that she finds acceptable, in order of preference. Further, each lecturer may have preferences over the students that find her projects acceptable and/or the projects that she offers. Typically there may be capacity constraint on the number of students that each project and lecturer can accommodate. The goal is to find a \emph{matching}, i.e., an assignment of students to projects based on the stated preferences such that each student is assigned to at most one project, and the capacity constraints on projects and lecturers are not violated.

Applications of {\sc spa} can be found in many university departments, for example, the School of Computing Science, University of Glasgow \cite{KIMS15}, the Faculty of Science, University of Southern Denmark \cite{CFG19}, the Department of Computing Science, University of York \cite{Kaz02}, and elsewhere \cite{AB03,RGSA17,HSVS05}. In this work, we will concern ourselves with a variant of {\sc spa} that involves lecturer preferences over students, which is known as the \emph{Student-Project Allocation problem with lecturer preferences over Students} ({\sc spa-s}) \cite{AIM07,Man13}. This variant falls under the category of bipartite matching problem with two-sided preferences.\footnote{For further reading on the classification of matching problems, we refer the interested reader to \cite{Man13}.}  In this context, it has been argued that a natural property for a matching to satisfy is that of \emph{stability} \cite{Rot84,Rot90,Rot91}. Informally, a \emph{stable matching} ensures that no student and lecturer would have an incentive to deviate from the matching by forming a private arrangement involving some project.

The classical {\sc spa-s} model assumes that preferences are strictly ordered. However, this might not be achievable in practice. For instance, a lecturer may be unable or unwilling to provide a strict ordering of all the students who find her projects acceptable. Such a lecturer may be happier to rank two or more students equally in a tie, which indicates that the lecturer is indifferent between the students concerned. This leads to a generalisation of {\sc spa-s} which we refer to as the \emph{Student-Project Allocation problem with lecturer preferences over Students with Ties} ({\sc spa-st}).

If we allow ties in the preference lists of students and lecturers, three different stability definitions naturally arise. Suppose $M$ is a matching in an instance of {\sc spa-st}. Informally, we say that $M$ is \emph{weakly stable, strongly stable} or \emph{super-stable} if there is no student and lecturer such that if they decide to form an arrangement outside the matching, respectively,
\begin{itemize}
\item[(i)] both of them would be better off, 
\item[(ii)] one of them would be better off and the other would be no worse off,
\item[(iii)] neither of them would be worse off.
\end{itemize}

With respect to this informal definition, a super-stable matching is also strongly stable, and a strongly stable matching is also weakly stable. These concepts were first defined and studied by Irving \cite{Irv94} in the context of the \emph{Stable Marriage problem with Ties} ({\sc smt}), and subsequently extended to the \emph{Hospitals/Residents problem with Ties} ({\sc hrt}) \cite{IMS00,IMS03} (where {\sc hrt} is the special case of {\sc spa-st} in which each lecturer offers only one project, and the capacity of each project is the same as the capacity of the lecturer offering the project; and {\sc smt} is a restriction of {\sc hrt} where the capacity of each hospital is $1$).

Considering the weakest of the three stability concepts mentioned above, every instance of {\sc spa-st} admits a weakly stable matching (this follows by breaking the ties in an arbitrary fashion and applying the stable matching algorithm described in \cite{AIM07} to the resulting {\sc spa-s} instance). However, such matchings could be of different sizes \cite{MIIMM02}. Thus opting for weak stability leads to the problem of finding a weakly stable matching that matches as many students to projects as possible -- a problem that is known to be NP-hard \cite{IMMM99,MIIMM02}, even for the so-called \emph{Stable Marriage problem with Ties and Incomplete lists} ({\sc smti}), which is an extension of {\sc smt} in which the preference lists need not be complete. However, we note that a $\frac{3}{2}$-approximation algorithm was described in \cite{CM18} for the problem of finding a maximum size weakly stable matching, given an instance of {\sc spa-st}.\footnote{This approximation algorithm finds a weakly stable matching that is at least two-thirds the size of a maximum weakly stable matching.}

Although a super-stable matching can be elusive, it avoids the problem of finding a maximum size weakly stable matching, because, as we will show in this paper, analogous to the {\sc hrt} case \cite{IMS00}: (i) all super-stable matchings have the same size; (ii) finding one or reporting that none exists can be accomplished in linear-time; and (iii) if a super-stable matching $M$ exists then all weakly stable matchings are of the same size (equal to the size of $M$), and match exactly the same set of students.  Furthermore, Irving \emph{et al}.~\cite{IMS00} argued that super-stability is a very natural solution concept in cases where agents have incomplete information.  Central to their argument is the following proposition, stated for {\sc hrt} in \cite[Proposition 2]{IMS00}, which extends naturally to {\sc spa-st} as follows (see Section \ref{subsection:spa-st} for a proof).
\begin{restatable}[]{proposition}{superstability}
\label{proposition1}
Let $I$ be an instance of {\sc spa-st}, and let $M$ be a matching in $I$. Then $M$ is super-stable in $I$ if and only if $M$ is stable in every instance of {\sc spa-s} obtained from $I$ by breaking the ties in some way.
\end{restatable}
In a practical setting, suppose that a student $s_i$ has incomplete information about two or more projects and decides to rank them equally in a tie $T$, and a super-stable matching $M$ exists in the corresponding {\sc spa-st} instance $I$. Then $M$ is stable in every instance of {\sc spa-s} (obtained from $I$ by breaking the ties) that represents the true preferences of $s_i$. 
Consequently, we will focus on the concept of super-stability in the {\sc spa-st} context.

Unfortunately not every instance of {\sc spa-st} admits a super-stable matching. This is true, for example, in the case where there are two students, two projects and one lecturer, the capacity of each project is $1$, the capacity of the lecturer is $2$, and every preference list is a single tie of length 2; any matching will be undermined by some student $s_i$ and the lecturer involving a project that $s_i$ is not assigned to. Nonetheless, it should be clear from the discussions above that a super-stable matching should be preferred in practical applications when one does exist.

\paragraph{\textbf{Related work.}} Irving \emph{et al}.~\cite{IMS00} described an algorithm to find a super-stable matching given an instance of {\sc hrt}, or to report that no such matching exists. However, merely reducing an instance of {\sc spa-st} to an instance of {\sc hrt} and applying the algorithm described in \cite{IMS00} to the resulting {\sc hrt} instance does not work in general (we explain this further in Section \ref{subsect:cloning}). Other variants of {\sc spa} in the literature involve lecturer preferences over their proposed projects \cite{IMY12,MMO18,MO08}, lecturer preferences over (student, project) pairs \cite{AM09}, and no lecturer preferences at all \cite{KIMS15} (see \cite{CFG19} for a more detailed survey in this latter case). A similar model known as the \textit{Student-Project-Resource Matching-Allocation problem} ({\sc spr}) was recently considered in \cite{IYYY19}. This model is different from {\sc spa-s} in the following ways: (i) in {\sc spa-s}, the capacity of each project is fixed by the lecturer offering it, while in {\sc spr}, the capacity of each project is determined by the resources allocated to it; (ii) in {\sc spa-s}, each lecturer has a fixed capacity on the total number of students that can be assigned to her projects, while in {\sc spr}, there is no notion of lecturer capacity. 

\paragraph{\textbf{Our contribution.}} In this paper, we describe the first polynomial-time algorithm to find a super-stable matching or to report that no such matching exists, given an instance of {\sc spa-st}  -- thus solving an open problem given in \cite{AIM07,Man13}. Our algorithm is student-oriented because it involves the students applying to projects. Moreover, the algorithm returns the student-optimal super-stable matching, in the sense that if the given instance admits a super-stable matching then our algorithm will output a solution in which each assigned student has the best project that she could obtain in any super-stable matching that the instance admits. We also present the results of an empirical evaluation based on an implementation of our algorithm that investigates how the nature of the preference lists would affect the likelihood of a super-stable matching existing, with respect to randomly-generated {\sc spa-st} instances.\footnote{From a theoretical perspective, the likelihood of a stable matching existing has been explored for the Stable Roommates problem -- a non-bipartite generalisation of the Stable Marriage problem \cite{PI94}.} Our main finding from the empirical evaluation is that super-stable matchings are very elusive with ties in the students' and lecturers' preference lists. However, if the preference lists of the students are strictly ordered and only the lecturers express ties in their preference lists, the probability of a super-stable matching existing is significantly higher.

The remainder of this paper is structured as follows. We give a formal definition of the {\sc spa-s} problem, the {\sc spa-st} variant, and the super-stability concept in Section \ref{section:definitions}. We describe our algorithm for {\sc spa-st} under super-stability in Section \ref{section:algorithm}. Further, Section \ref{section:algorithm} also presents our algorithm's correctness results and some structural properties satisfied by the set of super-stable matchings in an instance of {\sc spa-st}. In Section \ref{emprical-results}, we present the experimental results obtained from our algorithm's empirical evaluation.  Finally, Section \ref{section:conclusions} presents some concluding remarks and potential direction for future work.
\section{Preliminary definitions and results}
\label{section:definitions}
\subsection{Formal definition of {\footnotesize SPA-S}}
\label{subsection:spa-s}
An instance $I$ of {\sc spa-s} involves a set $\mathcal{S} = \{s_1 , s_2, \ldots , s_{n_1}\}$ of \emph{students}, a set $\mathcal{P} = \{p_1 , p_2, \ldots , p_{n_2}\}$ of \emph{projects} and a set $\mathcal{L} = \{l_1 , l_2, \ldots , l_{n_3}\}$ of \emph{lecturers}. Each student $s_i$ ranks a subset of $\mathcal{P}$ in strict order, which forms $s_i$'s preference list. We say that $s_i$ finds $p_j$ \emph{acceptable} if $p_j$ is in $s_i$'s preference list, and we denote by $A_i$ the set of projects that $s_i$ finds acceptable. Each lecturer $l_k \in \mathcal{L}$ offers a non-empty set of projects $P_k$, where $P_1, P_2, \ldots,$ $P_{n_3}$ partitions $\mathcal{P}$. Also, $l_k$ ranks in strict order of preference those students who find at least one project in $P_k$ acceptable, which forms $l_k$'s preference list. We say that $l_k$ finds $s_i$ \textit{acceptable} if $s_i$ is in $l_k$'s preference list, and we denote by $\mathcal{L}_k$ the set of students that $l_k$ finds acceptable. 

For any pair $(s_i, p_j) \in \mathcal{S} \times \mathcal{P}$, where $p_j$ is offered by $l_k$, we refer to $(s_i, p_j)$ as an \textit{acceptable pair} if $s_i$ and $l_k$ both find each other acceptable, i.e., if $p_j \in A_i$ and $s_i \in \mathcal{L}_k$. Each project $p_j \in \mathcal{P}$ has a capacity $c_j \in \mathbb{Z}^+$ indicating the maximum number of students that can be assigned to $p_j$. Similarly, each lecturer $l_k \in \mathcal{L}$ has a capacity $d_k \in \mathbb{Z}^+$ indicating the maximum number of students that $l_k$ is willing to supervise. We assume that for any lecturer $l_k$,

$$\max\{c_j: p_j \in P_k\} \leq d_k \leq \sum \{c_j: p_j \in P_k\},$$ 
\noindent
i.e., the capacity of $l_k$ is (i) at least the highest capacity of the projects offered by $l_k$, and (ii) at most the sum of the capacities of all the projects $l_k$ is offering. We denote by $\mathcal{L}_k^j$, the \emph{projected preference list} of lecturer $l_k$ for $p_j$, which can be obtained from $\mathcal{L}_k$ by removing those students that do not find $p_j$ acceptable (thereby retaining the order of the remaining students from $\mathcal{L}_k$).

An \emph{assignment} $M$ is a subset of $\mathcal{S} \times \mathcal{P}$ such that $(s_i, p_j) \in M$ implies that $s_i$ finds $p_j$ acceptable. If $(s_i, p_j) \in M$, we say that $s_i$ \emph{is assigned to} $p_j$, and $p_j$ \emph{is assigned} $s_i$. For convenience, if $s_i$ is assigned in $M$ to $p_j$, where $p_j$ is offered by $l_k$, we may also say that $s_i$ \emph{is assigned to} $l_k$, and $l_k$ \emph{is assigned} $s_i$.

For any student $s_i \in \mathcal{S}$, we let $M(s_i)$ denote the set of projects that are assigned to $s_i$ in $M$.  For any project $p_j \in \mathcal{P}$, we denote by $M(p_j)$ the set of students that are assigned to $p_j$ in $M$. Project $p_j$ is \emph{undersubscribed}, \emph{full} or \emph{oversubscribed} in $M$ according as $|M(p_j)|$ is less than, equal to, or greater than $c_j$, respectively. Similarly, for any lecturer $l_k \in \mathcal{L}$, we denote by $M(l_k)$ the set of students that are assigned to $l_k$ in $M$. Lecturer $l_k$ is \emph{undersubscribed}, \emph{full} or \emph{oversubscribed} in $M$ according as $|M(l_k)|$ is less than, equal to, or greater than $d_k$, respectively. 

A \emph{matching} $M$ is an assignment such that each student is assigned to at most one project in $M$, each project is assigned at most $c_j$ students in $M$, and each lecturer is assigned at most $d_k$ students in $M$ (i.e., $|M(s_i)| \leq 1$ for each $s_i \in \mathcal{S}$, $|M(p_j)| \leq c_j$ for each $p_j \in \mathcal{P}$, and $|M(l_k)| \leq d_k$ for each $l_k \in \mathcal{L}$). If $s_i$ is assigned to some project in $M$, for convenience we let $M(s_i)$ denote that project. In what follows, $l_k$ is the lecturer who offers project $p_j$.

\begin{definition}[Stability]
\label{def:stability}
Let $I$ be an instance of {\sc spa-st}, and let $M$ be a matching in $I$. We say that $M$ is \emph{stable} if it admits no blocking pair, where a \emph{blocking pair} is an acceptable pair $(s_i, p_j) \in (\mathcal{S} \times \mathcal{P}) \setminus M$ such that (a) and (b) holds as follows:
\begin{enumerate}[(a)]
    \item either $s_i$ is unassigned in $M$ or $s_i$ prefers $p_j$ to $M(s_i)$;
    \item either (i), (ii) or (iii) holds as follows:
    \begin{enumerate} [(i)]
\item each of $p_j$ and $l_k$ is undersubscribed in $M$;
\item $p_j$ is undersubscribed in $M$, $l_k$ is full in $M$ and either 

\begin{enumerate}[(1)]
    \item $s_i \in M(l_k)$, or
    \item $l_k$  prefers $s_i$ to the worst student in $M(l_k)$;
\end{enumerate}
\item $p_j$ is full in $M$ and $l_k$  prefers $s_i$ to the worst student in $M(p_j)$.
\end{enumerate}
\end{enumerate}

\end{definition}
To find a stable matching in an instance of {\sc spa-s}, two linear-time algorithms were described in \cite{AIM07}. The stable matching produced by the first algorithm is \emph{student-optimal} (i.e., each assigned student has the best-possible project that she could obtain in any stable matching) while the one produced by the second algorithm is \emph{lecturer-optimal} (i.e., each lecturer has the best set of students that she could obtain in any stable matching). The set of stable matchings in a given instance of {\sc spa-s} satisfy several interesting properties that together form what we will call the \emph{Unpopular Projects Theorem} (analogous to the Rural Hospitals Theorem for {\scriptsize HR} \cite{IMS00}), which we state as follows.
\begin{theorem}[\cite{AIM07}]
\label{thrm:rural-spa-s}
For a given instance of {\sc spa-s}, the following holds:
\begin{enumerate}
\item each lecturer is assigned the same number of students in all stable matchings;
\item exactly the same students are unassigned in all stable matchings;
\item a project offered by an undersubscribed lecturer is assigned the same number of students in all stable matchings.
\end{enumerate}
\end{theorem}

As we will see later in this paper, when ties are present in the preference lists of students and lecturers, the set of super-stable matchings also satisfy each of the properties in Theorem \ref{thrm:rural-spa-s}.

\subsection{Ties in the preference lists}
\label{subsection:spa-st}
We now define formally the generalisation of {\sc spa-s} in which the preference lists can include ties. In the preference list of lecturer $l_k\in \mathcal{L}$, a set $T$ of $r$ students forms a \emph{tie of length $r$} if $l_k$ does not prefer $s_i$ to $s_{i'}$ for any $s_i, s_{i'} \in T$ (i.e., $l_k$ is \emph{indifferent} between $s_i$ and $s_{i'}$). A tie in a student's preference list is defined similarly. For convenience, henceforth, we consider a non-tied entry in a preference list as a tie of length one. We denote by {\sc spa-st} the generalisation of {\sc spa-s} in which the preference list of each student (respectively lecturer) comprises a strict ranking of ties, each comprising one or more projects (respectively students). 

An example {\sc spa-st} instance $I_1$ is given in Fig.~\ref{fig:spa-st-instance-1}, which involves the set of students $\mathcal{S} = \{s_1, s_2, s_3, s_4, \\ s_5\}$, the set of projects $\mathcal{P} = \{p_1, p_2, p_3\}$ and the set of lecturers $\mathcal{L} = \{l_1, l_2\}$, with $P_1 = \{p_1, p_2\}$ and $P_2 = \{p_3\}$.  Ties in the preference lists are indicated by round brackets. 

\begin{figure}[t]
\centering
\small
\begin{tabular}{llll}
\hline
Students' preferences & \qquad \qquad  & Lecturers' preferences & offers\\ 
$s_1$: \;  $p_1$  &  & $l_1$: \; $s_5$ \;($s_1$ \; $s_2$) \; $s_3$ \; $s_4$ & $p_1$, $p_2$\\ 
$s_2$: \;($p_1$ \; $p_3$)  &  & $l_2$: \;  $s_4$ \; $s_5$ \; $s_2$ & $p_3$\\ 
$s_3$: \; $p_2$ &  & &\\
$s_4$: \; $p_2$ \; $p_3$ &  & Project capacities: $c_1 = c_3 = 1, \; c_2 = 2$& \\
$s_5$: \; $p_3$ \; $p_1$ &  & Lecturer capacities: $d_1 = 2, \; d_2 = 1$&\\ 
\hline
\end{tabular}
\caption{\label{fig:spa-st-instance-1} \small An example instance $I_1$ of {\sc spa-st}.}
\end{figure}

In the context of {\sc spa-st}, we assume that all notation and terminology carries over from Section \ref{subsection:spa-s} as defined for {\sc spa-s} with the exception of stability, which we now define. When ties appear in the preference lists, three levels of stability arise (as in the {\sc hrt} context \cite{IMS00,IMS03}), namely \emph{weak stability, strong stability and super-stability}. The formal definition for weak stability in {\sc spa-st} follows from the definition for stability in {\sc spa-s} (see Definition \ref{def:stability}). Moreover, the existence of a weakly stable matching in an instance $I$ of {\sc spa-st} is guaranteed by breaking the ties in $I$ arbitrarily, thus giving rise to an instance $I'$ of {\sc spa-s}. Clearly, a stable matching in $I'$ is weakly stable in $I$. Indeed a converse of sorts holds, which gives rise to the following proposition.

\begin{restatable}[]{proposition}{weakstability}
\label{proposition2}
Let $I$ be an instance of {\sc spa-st}, and let $M$ be a matching in $I$. Then $M$ is weakly stable in $I$ if and only if $M$ is stable in some instance $I'$ of {\sc spa-s} obtained from $I$ by breaking the ties in some way.
\end{restatable}
\begin{proof}
Let $I$ be an instance of {\sc spa-st} and let $M$ be a matching in $I$. Suppose that $M$ is weakly stable in $I$. Let $I'$ be an instance of {\sc spa-s} obtained from $I$ by breaking the ties in the following way. For each student $s_i$ in $I$ such that the preference list of $s_i$ includes a tie $T$ containing two or more projects, we order the preference list of $s_i$ in $I'$ as follows: if $s_i$ is assigned in $M$ to a project $p_j$ in $T$ then $s_i$ prefers $p_j$ to every other project in $T$; otherwise, we order the projects in $T$ arbitrarily. For each lecturer $l_k$ in $I$ such that $l_k$'s preference list includes a tie $X$, if $X$ contains students that are assigned to $l_k$ in $M$ and students that are not assigned to $l_k$ in $M$ then $l_k$'s preference list in $I'$ is ordered in such a way that each $s_i\in X\cap M(l_k)$ is preferred to each $s_{i'}\in X\setminus  M(l_k)$; otherwise, we order the students in $X$ arbitrarily. Now, suppose $(s_i, p_j)$ forms a blocking pair for $M$ in $I'$.  Given how the ties in $I$ were removed to obtain $I'$, this implies that $(s_i, p_j)$ forms a blocking pair for $M$ in $I$, a contradiction to our assumption that $M$ is weakly stable in $I$. Thus $M$ is stable in $I'$.

Conversely, suppose $M$ is stable in some instance $I'$ of {\sc spa-s} obtained from $I$ by breaking the ties in some way. Now suppose that $M$ is not weakly stable in $I$.  Then some pair $(s_i, p_j)$ forms a blocking pair for $M$ in $I$.  It is then clear from the definition of weak stability and from the construction of $I'$ that $(s_i, p_j)$ is a blocking pair for $M$ in $I'$, a contradiction. 
\qed \end{proof}
\noindent
As mentioned earlier, super-stability is the most robust concept to seek. Only if no super-stable matching exists in the underlying problem instance should other forms of stability be sought in a practical setting. Thus, for the remainder of this paper, we focus on super-stability in the {\sc spa-st} context.

\begin{definition}[Super-stability]
\label{definition:super-stability}
Let $I$ be an instance of {\sc spa-st}, and let $M$ be a matching in $I$. We say that $M$ is \emph{super-stable} if it admits no blocking pair, where a \emph{blocking pair} is an acceptable pair $(s_i, p_j) \in (\mathcal{S} \times \mathcal{P}) \setminus M$ such that (a) and (b) holds as follows:
\begin{enumerate}[(a)]
    \item either $s_i$ is unassigned in $M$ or $s_i$ prefers $p_j$ to $M(s_i)$ or is indifferent between them;
    \item either (i), (ii), or (iii) holds as follows:
    \begin{enumerate} [(i)]
\item each of $p_j$ and $l_k$ is undersubscribed in $M$;
\item $p_j$ is undersubscribed in $M$, $l_k$ is full in $M$ and either 

\begin{enumerate}[(1)]
    \item $s_i \in M(l_k)$, or
    \item $l_k$  prefers $s_i$ to the worst student/s in $M(l_k)$ or is indifferent between them;
\end{enumerate}
\item $p_j$ is full in $M$ and $l_k$  prefers $s_i$ to the worst student/s in $M(p_j)$ or is indifferent between them.
\end{enumerate}
\end{enumerate}

\end{definition}
It may be verified that the matching $M = \{(s_3, p_2), (s_4, p_3), (s_5, p_1)\}$ is super-stable in Fig.~\ref{fig:spa-st-instance-1}. Clearly, a super-stable matching is also weakly stable. Moreover, the super-stability definition gives rise to Proposition \ref{proposition1}, which can be regarded as an analogue of Proposition \ref{proposition2} for super-stability, restated as follows.
\superstability*

\begin{proof}
Let $I$ be an instance of {\sc spa-st} and let $M$ be a matching in $I$. Suppose that $M$ is super-stable in $I$. We want to show that $M$ is stable in every instance of {\sc spa-s} obtained from $I$ by breaking the ties in some way. Now, let $I'$ be an arbitrary instance of {\sc spa-s} obtained from $I$ by breaking the ties in some way, and suppose $M$ is not stable in $I'$. This implies that $M$ admits a blocking pair $(s_i, p_j)$ in $I'$. Since $I'$ is an arbitrary {\sc spa-s} instance obtained from $I$ by breaking the ties in some way, it follows that in $I$: (i) if $s_i$ is assigned in $M$ then $s_i$ either prefers $p_j$ to $M(s_i)$ or is indifferent between them, (ii) if $p_j$ is full in $M$ then $l_k$ either prefers $s_i$ to a worst student in $M(p_j)$ or is indifferent between them, and (iii) if $l_k$ is full in $M$ then either $s_i \in M(l_k)$ or $l_k$ prefers $s_i$ to a worst student in $M(l_k)$ or is indifferent between them. This implies that $(s_i, p_j)$ forms a blocking pair for $M$ in $I$, a contradiction to the super-stability of $M$. 

Conversely, suppose $M$ is stable in every instance of {\sc spa-s} obtained from $I$ by breaking the ties in some way. Now suppose $M$ is not super-stable in $I$. This implies that $M$ admits a blocking pair $(s_i, p_j)$ in $I$. We construct an instance $I'$ of {\sc spa-s} from $I$ by breaking the ties in the following way: (i) if $s_i$ is assigned in $M$ and $s_i$ is indifferent between $p_j$ and $M(s_i)$ in $I$ then $s_i$ prefers $p_j$ to $M(s_i)$ in $I'$; otherwise we break the ties in $s_i$'s preference list arbitrarily, and (ii) if some student, say $s_{i'}$, different from $s_i$ is assigned to $l_k$ in $M$ such that $l_k$ is indifferent between $s_i$ and $s_{i'}$ in $I$ then $l_k$ prefers $s_i$ to $s_{i'}$ in $I'$; otherwise we break the ties in $l_k$'s preference list arbitrarily. Thus $(s_i, p_j)$ forms a blocking pair for $M$ in $I'$, i.e., $M$ is not stable in $I'$, a contradiction to the fact that $M$ is stable in every instance of {\sc spa-s} obtained from $I$ by breaking the ties in some way.
\qed\end{proof}
\noindent 
The following proposition, which is a consequence of Propositions \ref{proposition1} and \ref{proposition2}, and Theorem \ref{thrm:rural-spa-s}, tells us that if a super-stable matching $M$ exists in $I$ then all weakly stable matchings in $I$ are of the same size (equal to the size of $M$) and match exactly the same set of students.
\begin{restatable}[]{proposition}{allinone}
\label{proposition3}
Let $I$ be an instance of {\sc spa-st}, and suppose that $I$ admits a super-stable matching $M$. Then the Unpopular Projects Theorem holds for the set of weakly stable matchings in $I$.
\end{restatable}

\begin{proof}
Let $I$ be an instance of {\sc spa-st}. Let $M$ be a super-stable matching in $I$ and let $M'$ be a weakly stable matching in $I$. Then by Proposition \ref{proposition2}, $M'$ is stable in some instance $I'$ of {\sc spa-s} obtained from $I$ by breaking the ties in some way. Also $M$ is stable in $I'$ by Proposition \ref{proposition1}. By Theorem \ref{thrm:rural-spa-s}, each lecturer is assigned the same number of students in $M$ and $M'$, exactly the same students are unassigned in $M$ and $M'$, and a project offered by an undersubscribed lecturer is assigned the same number of students in $M$ and $M'$. Hence, the Unpopular Projects Theorem holds for the set of weakly stable matchings in $I$. \qed
\end{proof}

\subsection{Cloning from {\sc spa-st} to {\sc hrt} does not work in general}
\label{subsect:cloning}
As mentioned earlier, Irving \textit{et al.}~\cite{IMS00} described a polynomial-time algorithm to find a super-stable matching or report that no such matching exists, given an instance of {\sc hrt}. The authors referred to their algorithm as Algorithm {\sf HRT-Super-Res}. One might assume that reducing a given instance of {\sc spa-st} to an instance of {\sc hrt} (using a ``cloning'' technique) and subsequently applying Algorithm {\sf HRT-Super-Res} to the resulting instance would solve our problem. However, this is not always true. In what follows, we describe an obvious method to clone an instance of {\sc spa-st} to an instance of {\sc hrt}, and we show that applying the super-stable matching algorithm described in \cite{IMS00} to the resulting {\sc hrt} instance does not work in general.

A method to derive an instance $I'$ of {\sc hrt} from an instance $I$ of {\sc spa-st} was described by Cooper and Manlove \cite{CM18a}. We explain this method as follows. The students and projects involved in $I$ are converted into residents and hospitals respectively in $I'$, i.e., each $s_i \in \mathcal{S}$ becomes $r_i$ in the cloned instance, and each $p_j \in \mathcal{P}$ becomes $h_j$. Residents inherit their preference lists naturally from students, i.e., if $r_i$ corresponds to $s_i$ then the preference list of $r_i$ in $I'$ is $A_i$, with each project in $A_i$ being replaced by the associated hospital. Hospitals inherit their preference lists from the projected preference list of the associated project according to the lecturer offering the project, i.e., if $p_j$ corresponds to $h_j$ (where $p_j$ is offered by $l_k$) then the preference list of $h_j$ in $I'$ is $\mathcal{L}_k^j$, with each student in $\mathcal{L}_k^j$ being replaced by the associated resident. Each hospital also inherits its capacity from the project, i.e., for each $h_j$ associated with $p_j$, the capacity of $h_j$ is $c_j$. 

Let $l_k$ be an arbitrary lecturer in $I$. In order to translate $l_k$'s capacity into the {\sc hrt} instance, we create $n$ \emph{dummy residents}\footnote{The dummy residents created for each hospital will offset the difference between the corresponding lecturer capacity and the total capacity of her proposed projects.} for each hospital $h_j$ corresponding to a project $p_j \in P_k$, where $n$ is the difference between the sum of the capacities of all the projects in $P_k$ and the capacity of $l_k$ (recall that $\sum_{p_j \in P_k} c_j \geq d_k$). The preference list for each of these dummy residents will be a single tie consisting of all the hospitals corresponding to a project in $P_k$. Further, the preference list for each hospital corresponding to a project in $P_k$ will include a tie in its first position consisting of all the dummy residents associated with $l_k$. 

Next, we describe how to map between matchings in $I$ and in $I'$. Let $M$ and $M'$ be a matching in $I$ and $I'$ respectively. Let $r_i$ be the resident associated with $s_i$ and let $h_j$ be the hospital associated with $p_j$. If $s_i$ is assigned in $M$ to project $p_j$, then $r_i$ is assigned in $M'$ to hospital $h_j$. To illustrate the cloning technique described above, we give an example instance $I$ of {\sc spa-st} in Fig.~\ref{fig:super-instance-2} as well as the corresponding cloned {\sc hrt} instance $I'$ in Fig.~\ref{fig:super-instance-2-cloned}. Also, we give an intuition as to why this technique will not work in general.  

\begin{figure}[t]
\centering
\renewcommand{\arraystretch}{1}
\begin{tabular}{llll}
\hline
Students' preferences & \qquad \qquad   & Lecturers' preferences &  offers \\ 
$s_1$:\; $p_1$  &  & $l_1$: \;$s_1$\;($s_2$ \; $s_3$)  & $p_1$, $p_2$\\ 
$s_2$:\;($p_1$ \; $p_2$)  &  & $l_2$: \;$s_3$ & $p_3$\\ 
$s_3$:\; $p_2$ \; $p_3$  &   & &\\  
 &  & Project capacities: $c_1 = c_2 = c_3 = 1$& \\
 &  & Lecturer capacities: $d_1 = d_2 = 1$&\\ 
\hline
\end{tabular}
\caption[]{\label{fig:super-instance-2} An instance $I$ of {\sc spa-st}.}
\end{figure}

\begin{figure}[t]
\centering
\mbox{} \\
\mbox{} \\
\renewcommand{\arraystretch}{1}
\begin{tabular}{lll}
\hline
Residents' preferences &  \qquad \qquad  & Hospitals' preferences\\ 
$r_1$: \; $h_1$  &  & $h_1$: \;$r_{d_1}$ \; $r_1$ \; $r_2$\\ 

$r_2$: \;($h_1$ \; $h_2$)    &  & $h_2$: \;$r_{d_1}$ \;($r_2$ \; $r_3$)\\  

$r_3$: \; $h_2$ \; $h_3$ &  &$h_3$:\; $r_3$\\  
$r_{d_1}$:\;($h_1$ \; $h_2$)  &  &\\  
 &  & Hospital capacities: $c_1 = c_2 = c_3 = 1$\\  
\hline
\end{tabular}
\caption[]{\label{fig:super-instance-2-cloned} An instance $I'$ of {\sc hrt} cloned from the {\sc spa-st} instance illustrated in Fig.~\ref{fig:super-instance-2}.}

\end{figure}

With respect to Figs.~\ref{fig:super-instance-2} and \ref{fig:super-instance-2-cloned}, each resident $r_1, r_2$ and $r_3$ in $I'$ corresponds to student $s_1, s_2$ and $s_3$ in $I$, respectively; and the preference list of each resident is adapted from the preference list of the associated student. Also, each hospital $h_1, h_2$ and $h_3$ in $I'$ corresponds to project $p_1, p_2$ and $p_3$ in $I$, respectively. The preference list of hospitals $h_1$ and $h_2$ is $\mathcal{L}_1^1$ and $\mathcal{L}_1^2$ respectively, since $l_1$ is the lecturer that offers both $p_1$ and $p_2$. Similarly, the preference list of hospital $h_3$ is $\mathcal{L}_2^3$, since $l_2$ is the lecturer that offers $p_3$. Further, for lecturer $l_1$ who offers both $p_1$ and $p_2$, since $c_1 + c_2 = 2 > 1 = d_1$, we add one dummy resident $r_{d_1}$ to the cloned instance. The preference list of $r_{d_1}$ is a single tie consisting of $h_1$ and $h_2$; and the preference list of both $h_1$ and $h_2$ includes $r_{d_1}$ in first position.

The reader can easily verify that matching $M = \{(s_1, p_1), (s_3, p_3)\}$ is super-stable in the {\sc spa-st} instance $I$ illustrated in Fig.~\ref{fig:super-instance-2}. Now, following our description of how to map between matchings in $I$ and in $I'$, a matching in $I'$ is $M' = \{(r_{d_1}, h_2), (r_1, h_1), (r_3, h_3)\}$, with $(s_1, p_1) \in M$ corresponding to $(r_1, h_1) \in M'$ and  $(s_3, p_3) \in M$ corresponding to $(r_3, h_3) \in M'$. Clearly, $M'$ is not super-stable in $I'$ as $(r_{d_1}, h_1)$ forms a blocking pair. In fact, the {\sc hrt} instance $I'$ admits no super-stable matching. The justification for this is as follows: irrespective of the hospital that the dummy resident $r_{d_1}$ is assigned to in any matching obtained from $I'$, $r_{d_1}$ will block this matching via the other hospitals tied in her preference list (since the hospital would be better off taking on $r_{d_1}$, and $r_{d_1}$ would be no worse off).

One way to avoid this problem would be to strictly order the hospitals in $r_{d_1}$'s preference list; however, the order in which the hospitals appear will lead to different possibilities. For instance: if $r_{d_1}$ prefers $h_1$ to $h_2$, the reader can verify that the corresponding {\sc hrt} instance admits no super-stable matching; however, if $r_{d_1}$ prefers $h_2$ to $h_1$, again the reader can verify that the corresponding {\sc hrt} instance admits the super-stable matching $\{(r_{d_1}, h_2), (r_1, h_1), (r_3, h_3)\}$. The downside of this strategy is that there is no obvious reason as to why $r_{d_1}$ should prefer $h_2$ to $h_1$ in the cloned {\sc hrt} instance in Fig.~\ref{fig:super-instance-2-cloned} by merely looking at the original {\sc spa-st} instance in Fig.~\ref{fig:super-instance-2}. Hence, in order to make this technique work in general, we will need to generate every {\sc hrt} instance obtained by ordering the dummy residents' preference lists in some way. This is exponential in the problem instance.
\section{An algorithm for {\small SPA-ST} under super-stability}
\label{section:algorithm}

In this section we present our algorithm for {\sc spa-st} under super-stability, which we will refer to as Algorithm {\sf SPA-ST-super}. Before we proceed, we briefly describe Algorithm {\sf HRT-Super-Res} \cite{IMS00}. The algorithm involves a sequence of proposals from the residents to the hospitals. Each resident proposes in turn to all of the hospitals tied together at the head of her preference list, and all proposals are provisionally accepted. If a hospital $h$ becomes oversubscribed then none of $h$'s worst assignees nor any resident tied with these assignees in $h$'s preference list can be assigned to $h$ in any super-stable matching -- such pairs $(r, h)$ are deleted from each other's preference lists. If a hospital $h$ is full then no resident strictly worse than $h$'s worst assignees can be assigned to $h$ in any super-stable matching -- again such $(r,h)$ pairs are deleted from each other's preference lists. The proposal sequence terminates once every resident is either assigned to a hospital or has an empty preference list. At this point, if the constructed assignment of residents to hospitals is super-stable in the original {\sc hrt} instance then the assignment is returned as a super-stable matching. Otherwise, the algorithm reports that no super-stable matching exists.

We note that our algorithm is a non-trivial extension of Algorithm {\sf HRT-Super-Res} for {\sc hrt} \cite{IMS00}. Due to the more general setting of {\sc spa-st}, Algorithm {\sf SPA-ST-super} requires some new ideas (precisely lines 27-34 of the algorithm on page \pageref{algorithmSPA-STsuper}), and the proofs of the correctness results are more complex than for the aforementioned algorithm for {\sc hrt}. We give definitions relating to the algorithm in Section \ref{subsect:algorithm-definition}. We give a description of our algorithm in Section \ref{subsect:algorithm-description}, before presenting it in pseudocode form. In Section \ref{example-description}, we illustrate an execution of our algorithm with respect to an example {\sc spa-st} instance. We present the algorithm's correctness results in Section \ref{correctness-result}. Finally, in Section \ref{subsect:properties}, we show that the set of super-stable matchings in an instance of {\sc spa-st} satisfy analogous properties to those given in Theorem \ref{thrm:rural-spa-s}.
\subsection{Definitions relating to the algorithm}
\label{subsect:algorithm-definition}

First, we present some definitions relating to the algorithm. In what follows, $I$ is an instance of {\sc spa-st}, $(s_i, p_j)$ is an acceptable pair in $I$ and $l_k$ is the lecturer who offers $p_j$.  Further, if $(s_i,p_j)$ belongs to some super-stable matching in $I$, we call $(s_i, p_j)$ a \textit{super-stable pair}.

During the execution of the algorithm, students become \textit{provisionally assigned} to projects. It is possible for a project to be provisionally assigned a number of students that exceed its capacity. This holds analogously for a lecturer. The algorithm proceeds by deleting from the preference lists certain $(s_i, p_j)$ pairs that cannot be super-stable. By the term \textit{delete} $(s_i, p_j)$, we mean the removal of $p_j$ from $s_i$'s preference list and the removal of $s_i$ from $\mathcal{L}_k^j$ (the projected preference list of lecturer $l_k$ for $p_j$). In addition, if $s_i$ is provisionally assigned to $p_j$ at this point, we break the assignment.  If $s_i$ has been deleted from every projected preference list of $l_k$ that she originally belonged to, we will implicitly assume that $s_i$ has been deleted from $l_k$'s preference list.  By the \textit{head} of a student's preference list at a given point, we mean the set of one or more projects, tied in her preference list after any deletions might have occurred, that she prefers to all other projects in her list. 

For project $p_j$, we define the \textit{tail} of $\mathcal{L}_k^j$ as the least-preferred tie in $\mathcal{L}_k^j$ after any deletions might have occurred (recalling that a tie can be of length one). In the same fashion, we define the \textit{tail} of $\mathcal{L}_k$ (the preference list of lecturer $l_k$) as the least-preferred tie in $\mathcal{L}_k$ after any deletions might have occurred. If $s_i$ is provisionally assigned to $p_j$, we define the \textit{successors} of $s_i$ in $\mathcal{L}_{k}^j$ as those students that are worse than $s_i$ in $\mathcal{L}_{k}^j$. An analogous definition holds for the successors of $s_i$ in $\mathcal{L}_k$. 

\subsection{Description of the algorithm}
\label{subsect:algorithm-description}
We now describe our algorithm, shown in pseudocode form in Algorithm~\ref{algorithmSPA-STsuper}.  Algorithm {\sf SPA-ST-super} begins by initialising an empty set $M$ which will contain the provisional assignments of students to projects (and implicitly to lecturers). We remark that such assignments can subsequently be broken during the algorithm's execution. Also, each project is initially assigned to be empty (i.e., not assigned to any student).

The \texttt{while} loop of the algorithm involves each student $s_i$ who is not provisionally assigned to any project in $M$ and who has a non-empty preference list applying in turn to each project $p_j$ at the head of her list. Immediately, $s_i$ becomes provisionally assigned to $p_j$ in $M$ (and to $l_k$).
If, by gaining a new student, $p_j$ becomes oversubscribed, it turns out that none of the students $s_t$  at the tail of $\mathcal{L}_k^j$ can be assigned to $p_j$ in any super-stable matching -- such pairs $(s_t, p_j)$ are deleted. Similarly, if by gaining a new student, $l_k$ becomes oversubscribed, none of the students $s_t$ at the tail of $\mathcal{L}_k$ can be assigned to any project offered by $l_k$ in any super-stable matching  -- the pairs $(s_t, p_u)$, for each project $p_u \in P_k$ that $s_t$ finds acceptable, are deleted.

Regardless of whether any deletions occurred as a result of the two conditionals described in the previous paragraph, we have two further (possibly non-disjoint) cases in which deletions may occur. 
If $p_j$ becomes full, we let $s_r$ be any worst student provisionally assigned to $p_j$ (according to $\mathcal{L}_k^j$), and we delete $(s_t, p_j)$ for each successor $s_t$ of $s_r$ in $\mathcal{L}_k^j$. Similarly if $l_k$ becomes full, we let $s_r$ be any worst student provisionally assigned to $l_k$, and we delete $(s_t, p_u)$, for each successor $s_t$ of $s_r$ in $\mathcal{L}_k$ and for each project $p_u \in P_k$ that $s_t$ finds acceptable. As we will prove later, none of the (student, project) pairs that we delete is a super-stable pair.

At the point where the \texttt{while} loop terminates (i.e., when every student is provisionally assigned to one or more projects or has an empty preference list), if some project $p_j$ that was previously full ends up undersubscribed, we let $s_r$ be any one of the most-preferred students (according to $\mathcal{L}_k^j$) who was provisionally assigned to $p_j$ during some iteration of the algorithm but is not assigned to $p_j$ at this point (for convenience, we henceforth refer to such $s_r$ as the most-preferred student rejected from $p_j$ according to $\mathcal{L}_k^j$). If the students at the tail of $\mathcal{L}_k$ (recalling that the tail of $\mathcal{L}_k$ is the least-preferred tie in $\mathcal{L}_k$ after any deletions might have occurred) are no better than $s_r$, it turns out that none of these students $s_t$ can be assigned to any project offered by $l_k$ in any super-stable matching -- the pairs $(s_t, p_u)$, for each project $p_u \in P_k$ that $s_t$ finds acceptable, are deleted. The \texttt{while} loop is then potentially reactivated, and the entire process continues until every student is provisionally assigned to a project or has an empty preference list, at which point the \texttt{repeat-until} loop terminates.

Upon termination of the \texttt{repeat-until} loop, if the set $M$, containing the assignment of students to projects, is super-stable relative to the given instance $I$ then $M$ is output as a super-stable matching in $I$. Otherwise, the algorithm reports that no super-stable matching exists in $I$.
\begin{algorithm}[htbp]
\caption{Algorithm {\sf SPA-ST-super}}
\label{algorithmSPA-STsuper} 

\begin{algorithmic}[1]
\Require {{\sc spa-st} instance $I$} 
 \Ensure{a super-stable matching $M$ in $I$ or ``no super-stable matching exists in $I$''}
 
\State $M \gets \emptyset$
\ForEach {$p_j \in \mathcal{P}$}
\State  \texttt{full}($p_j$) = \texttt{false}
\EndFor
\Repeat{}
\While {some student $s_i$ is unassigned and has a non-empty preference list}
	\ForEach {project $p_j$ at the head of $s_i$'s preference list}
    	\State $l_k \gets $ lecturer who offers $p_j$
        \State /* $s_i$ applies to $p_j$ */
		\State $M \gets M \cup \{(s_i, p_j)\}$ /*provisionally assign $s_i$ to $p_j$ (and to $l_k$) */
        
		\If {$p_j$ is oversubscribed}
			\ForEach{student $s_t$ at the tail of $\mathcal{L}_{k}^{j}$}
				\State delete $(s_t, p_j)$ 

			\EndFor
		
		\ElsIf {$l_k$ is oversubscribed}
			\ForEach{student $s_t$ at the tail of $\mathcal{L}_{k}$}
				\ForEach {project $p_u \in  P_k \cap A_t$}
					\State delete $(s_t, p_u)$ 

				\EndFor
			\EndFor
			
		\EndIf
		
		\If {$p_j$ is full}
			\State \texttt{full}($p_j$) = \texttt{true}
			\State $s_r \gets $ worst student assigned to $p_j$ according to $\mathcal{L}_{k}^{j}$ \{any if $> 1$\}
			\ForEach{successor $s_t$ of $s_r$ on $\mathcal{L}_{k}^{j}$}
				 \State delete $(s_t, p_j)$
			\EndFor

		\EndIf
		
		\If {$l_k$ is full}
			\State $s_r \gets $ worst student assigned to $l_k$ according to $\mathcal{L}_{k}$ \{any if $> 1$\}
			\ForEach{successor $s_t$ of $s_r$ on $\mathcal{L}_{k}$}
				
				\ForEach{project $p_u \in P_k \cap A_t$ }
					 \State delete $(s_t, p_u)$ 
				\EndFor
			
			\EndFor

		\EndIf
	\EndFor
\EndWhile	
\ForEach{$p_j \in \mathcal{P}$}
\If {$p_j$ is undersubscribed and \texttt{full}($p_j$) is \texttt{true}}
 \State $l_k \gets $ lecturer who offers $p_j$
\State $s_r \gets $ most-preferred student rejected from $p_j$ according to $\mathcal{L}_{k}^{j}$ \{any if $> 1$\}
\If{the students at the tail of $\mathcal{L}_k$ are no better than $s_r$}
				 	\ForEach{student $s_t$ at the tail of $\mathcal{L}_k$}

						\ForEach{project $p_u \in P_k \cap A_t$ }
							\State delete $(s_t, p_u)$  \label{alg:deletion-outside}

						\EndFor
			
					\EndFor
				
	\EndIf			

		\EndIf	
\EndFor
\Until {every unassigned student has an empty preference list}
\If {$M$ is super-stable in $I$}
\State \Return $M$

\Else
 \State \Return ``no super-stable matching exists in $I$''

\EndIf
\end{algorithmic}
\end{algorithm}
\newpage
\subsection{Example algorithm execution}
\label{example-description}
We illustrate an  execution of Algorithm {\sf SPA-ST-super} with respect to the {\sc spa-st} instance shown in Fig.~\ref{fig:spa-st-instance-1} (page \pageref{fig:spa-st-instance-1}).  We initialise $M = \{\}$, which will contain the provisional assignment of students to projects. For each project $p_j \in \mathcal{P}$, we set \texttt{full}($p_j$) = \texttt{false} (\texttt{full}($p_j$) will be set to \texttt{true} when $p_j$ becomes full, so that we can easily identify any project that was full during an iteration of the algorithm and ended up undersubscribed). We assume that the students become provisionally assigned to each project at the head of their list in subscript order. Table~\ref{example-illustration} illustrates how this execution of Algorithm {\sf SPA-ST-super} proceeds with respect to $I_1$.
\begin{table}[htbp]
\caption{\label{example-illustration} \small An execution of Algorithm {\sf SPA-ST-super} with respect to Fig.~\ref{fig:spa-st-instance-1}.}
\centering \small
\setlength{\tabcolsep}{0.8em}
\renewcommand{\arraystretch}{1.7}
\begin{tabular}{p{1.6cm}p{2.4cm}p{10cm}}
\hline\noalign{\smallskip}
{\texttt while} loop iterations & Student applies to project & Consequence \\ 
\noalign{\smallskip}\hline\noalign{\smallskip}
$1$ & $s_1$ applies to $p_1$ & $M=\{(s_1, p_1)\}$. \texttt{full}($p_1$) = \texttt{true}. \\ 
\hline
$2$ & $s_2$ applies to $p_1$ & $M=\{(s_1, p_1), (s_2, p_1)\}$. $p_1$ becomes oversubscribed. The tail of $\mathcal{L}_1^1$ contains $s_1$ and $s_2$ -- thus we delete the pairs $(s_1, p_1)$ and $(s_2, p_1)$ (and we break the provisional assignments). \\ 

 & $s_2$ applies to $p_3$ &  $M=\{(s_2, p_3)\}$. \texttt{full}($p_3$) = \texttt{true}.\\
\hline
$3$ & $s_3$ applies to $p_2$ &  $M=\{(s_2, p_3), (s_3, p_2)\}$. \\
\hline
$4$ & $s_4$ applies to $p_2$ &  $M = \{(s_2, p_3), (s_3, p_2), (s_4, p_2)\}$. \texttt{full}($p_2$) = \texttt{true}.  \\
\hline
$5$ & $s_5$ applies to $p_3$ &  $M = \{(s_2, p_3), (s_3, p_2), (s_4, p_2), (s_5, p_3)\}$. $p_3$ becomes oversubscribed. The tail of $\mathcal{L}_2^3$ contains only $s_2$ -- thus we delete the pair $(s_2, p_3)$ (and we break the provisional assignment).\\
\hline
\multicolumn{3}{p{15.2cm}}{The first iteration of the \texttt{while} loop terminates since every unassigned student (i.e., $s_1$ and $s_2$) has an empty preference list. At this point, \texttt{full}($p_1$) is \texttt{true} and $p_1$ is undersubscribed. Moreover, the student at the tail of $\mathcal{L}_1$ (i.e., $s_4$) is no better than $s_1$, where $s_1$ was previously assigned to $p_1$ and $s_1$ is also the most-preferred student rejected from $p_1$ according to $\mathcal{L}_1^1$; thus we delete the pair $(s_4, p_2)$. The \texttt{while} loop is then reactivated.}\\
\hline
$6$ & $s_4$ applies to $p_3$ &  $M = \{(s_3, p_2), (s_5, p_3), (s_4, p_3)\}$. $p_3$ becomes oversubscribed. The tail of $\mathcal{L}_2^3$ contains only $s_5$ -- thus we delete the pair $(s_5, p_3)$.\\
\hline
$7$ & $s_5$ applies to $p_1$ &  $M = \{(s_3, p_2), (s_4, p_3), (s_5, p_1)\}$. \\
\hline
\multicolumn{3}{p{15.2cm}}{Again, every unassigned students has an empty preference list. We also have that \texttt{full}($p_2$) is \texttt{true} and $p_2$ is undersubscribed; however no further deletion is carried out in line 34 of the algorithm, since the student at the tail of $\mathcal{L}_1$ (i.e., $s_3$) is better than $s_4$, where $s_4$ was previously assigned to $p_2$ and $s_4$ is also the most-preferred student rejected from $p_2$ according to $\mathcal{L}_1^2$. Hence, the \texttt{repeat-until} loop terminates and the algorithm outputs $M = \{(s_3, p_2), (s_4, p_3), (s_5, p_1)\}$ as a super-stable matching. It is clear that $M$ is super-stable in the original instance $I_2$.}\\
\noalign{\smallskip}\hline
\end{tabular} 
\end{table}
\subsection{Correctness of Algorithm {\sf SPA-ST-super}}
\label{correctness-result}
We now present a series of results concerning the correctness of Algorithm {\sf SPA-ST-super}. The first of these results deals with the fact that no super-stable pair is deleted during an execution of the algorithm.
In what follows, $I$ is an instance of {\sc spa-st}, $(s_i, p_j)$ is an acceptable pair in $I$ and $l_k$ is the lecturer who offers $p_j$.
\begin{restatable}[]{lemma}{nopairdeletion}
\label{pair-deletion}
If a pair $(s_i, p_j)$ is deleted during an execution of Algorithm {\sf SPA-ST-super}, then $(s_i, p_j)$ does not belong to any super-stable matching in $I$.
\end{restatable}
\noindent
In order to prove Lemma \ref{pair-deletion}, we present Lemmas \ref{lemma:super-pair-deletion-within} and \ref{lemma:super-pair-deletion-outside}.
\begin{lemma}
\label{lemma:super-pair-deletion-within}
If a pair $(s_i, p_j)$ is deleted within the \texttt{while} loop during an execution of Algorithm {\sf SPA-ST-super} then $(s_i, p_j)$ does not belong to any super-stable matching in $I$.
\end{lemma}

\begin{proof}
Without loss of generality, suppose that the first super-stable pair to be deleted within the \texttt{while} loop during an arbitrary execution $E$ of the algorithm is $(s_i, p_j)$, which belongs to some super-stable matching, say $M^*$. Suppose that $M$ is the assignment immediately after the deletion. Let us denote this point in the algorithm where the deletion is made by $\ddagger$. During $E$, there are four cases that would lead to the deletion of any (student, project) pair within the \texttt{while} loop.
\begin{enumerate}[(1)]
\item \emph{$p_j$ is oversubscribed.} Suppose that $(s_i, p_j)$ is deleted because some student (possibly $s_i$) became provisionally assigned to $p_j$ during $E$, causing $p_j$ to become oversubscribed. If $p_j$ is full or undersubscribed at point $\ddagger$, since $s_{i} \in M^*(p_j) \setminus M(p_j)$ and no project can be oversubscribed in $M^*$, then there is some student $s_r \in M(p_j) \setminus M^*(p_j)$ such that $l_k$ prefers $s_r$ to $s_i$ or is indifferent between them. We note that $s_r$ cannot be assigned to a project that she prefers to $p_j$ in any super-stable matching. Otherwise, since $p_j$ must have been in the head of $s_r$'s preference list when she applied, this would mean that a super-stable pair was deleted before $(s_i, p_j)$. Thus either $s_r$ is unassigned in $M^*$ or $s_r$ prefers $p_j$ to $M^*(s_r)$ or $s_r$ is indifferent between them. Clearly, for any combination of $l_k$ and $p_j$ being full or undersubscribed in $M^*$, it follows that $(s_r, p_j)$ blocks $M^*$, a contradiction.
\item \emph{$l_k$ is oversubscribed.} Suppose that $(s_i, p_j)$ is deleted because some student (possibly $s_i$) became provisionally assigned to a project offered by lecturer $l_k$ during $E$, causing $l_k$ to become oversubscribed. At point $\ddagger$, none of the projects offered by $l_k$ is oversubscribed in $M$, otherwise we will be in case (1). 
Similar to case (1), if $l_k$ is full or undersubscribed at point $\ddagger$, since $s_{i} \in M^*(p_{j}) \setminus M(p_{j})$ and no lecturer can be oversubscribed in $M^*$, it follows that there is some project $p_{j'} \in P_k$ and some student $s_{r} \in M(p_{j'}) \setminus M^*(p_{j'})$ such that $l_k$ prefers $s_{r}$ to $s_i$ or is indifferent between them. We consider two subcases.
\begin{enumerate}[(i)]
\item If $p_{j'} = p_j$ then $s_{r} \neq s_i$. Moreover, as in case (1), either $s_{r}$ is unassigned in $M^*$ or $s_{r}$ prefers $p_{j'}$ to $M^*(s_{r})$ or $s_r$ is indifferent between them. For any combination of $l_k$ and $p_{j'}$ being full or undersubscribed in $M^*$, we have that $(s_{r}, p_{j'})$ blocks $M^*$, a contradiction.
\item If $p_{j'} \neq p_j$. Assume firstly that $s_{r} \neq s_i$. Then as $p_{j'}$ has fewer assignees in $M^*$ than it has provisional assignees in $M$, and as in (i) above, $(s_{r}, p_{j'})$ blocks $M^*$, a contradiction. Finally assume $s_{r} = s_i$. Then $s_i$ must have applied to $p_{j'}$ at some point during $E$ before $\ddagger$. Clearly, either $s_i$ prefers $p_{j'}$ to $p_j$ or $s_i$ is indifferent between them, since $p_{j'}$ must have been in the head of $s_i$'s preference list when $s_i$ applied. Since $s_i \in M^*(l_k)$ and $p_{j'}$ is undersubscribed in $M^*$, it follows that $(s_i, p_{j'})$ blocks $M^*$, a contradiction.
\end{enumerate}

\item \emph{$p_j$ is full.} Suppose that $(s_i, p_j)$ is deleted because $p_j$ became full during $E$. At point $\ddagger$, $p_j$ is full in $M$. Thus at least one of the students in $M(p_j)$, say $s_{r}$, will not be assigned to $p_j$ in $M^*$, for otherwise $p_j$ will be oversubscribed in $M^*$. This implies that either $s_{r}$ is unassigned in $M^*$ or $s_{r}$ prefers $p_j$ to $M^*(s_{r})$ or $s_{r}$ is indifferent between them. For otherwise, we obtain a contradiction to $(s_i, p_j)$ being the first super-stable pair to be deleted. Since $l_k$ prefers $s_{r}$ to $s_i$, it follows that $(s_{r}, p_j)$ blocks $M^*$, a contradiction.

\item \emph{$l_k$ is full.} Suppose that $(s_i, p_j)$ is deleted because $l_k$ became full during $E$. We consider two subcases.
\begin{enumerate}[(i)]
	\item All the students assigned to $p_j$ in $M$ at point $\ddagger$ (if any) are also assigned to $p_j$ in $M^*$. This implies that $p_j$ has one more assignee in $M^*$ than it has provisional assignees in $M$, namely $s_i$. Thus, some other project $p_{j'} \in P_k$ has fewer assignees in $M^*$ than it has provisional assignees in $M$, for otherwise $l_k$ would be oversubscribed in $M^*$. Hence there exists some student $s_{r} \in M(p_{j'}) \setminus M^*(p_{j'})$.  It is clear that $s_{r} \neq s_i$, since $s_i$ plays the role of $s_t$ at some for loop iteration in line 24 of the algorithm.
	Also, $s_{r}$ cannot be assigned to a project that she prefers to $p_{j'}$ in $M^*$, as explained in case (1). Moreover, since $p_{j'}$ is undersubscribed in $M^*$ and $l_k$ prefers $s_{r}$ to $s_i$, it follows that $(s_{r}, p_{j'})$ blocks $M^*$, a contradiction. 
	\item Some student, say $s_{r}$, who is assigned to $p_j$ in $M$ is not assigned to $p_j$ in $M^*$, i.e., $s_{r} \in M(p_j) \setminus M^*(p_j)$. Since $s_{r}$ cannot be assigned in $M^*$ to a project that she prefers to $p_j$ and since $l_k$ prefers $s_{r}$ to $s_i$, it follows that $(s_{r}, p_j)$ blocks $M^*$, a contradiction.
\end{enumerate}
\end{enumerate} \qed \end{proof}
\begin{lemma}
\label{lemma:super-pair-deletion-outside}
If a pair $(s_i, p_j)$ is deleted in line 34 of Algorithm {\sf SPA-ST-super} then $(s_i, p_j)$ does not belong to any super-stable matching in $I$.
\end{lemma}
\begin{proof}
Without loss of generality, suppose that the first super-stable pair to be deleted during an arbitrary execution $E$ of the algorithm is $(s_i, p_j)$, which belongs to some super-stable matching, say $M^*$.  Then by Lemma \ref{lemma:super-pair-deletion-within}, $(s_i, p_j)$ was deleted in line 34 during $E$. Let $l_k$ be the lecturer who offers $p_j$. Suppose that $M$ is the assignment during the iteration of the \texttt{repeat-until} loop where $(s_i, p_j)$ was deleted.

Let $p_{j'}$ be some other project offered by $l_k$ which was full during a previous \texttt{repeat-until} loop iteration and subsequently ends up undersubscribed in the current \texttt{repeat-until} loop iteration, i.e., $p_{j'}$ plays the role of $p_j$ in line 28. Suppose that $s_{i'}$ plays the role of $s_r$ in line 30, i.e., $s_{i'}$ is the most-preferred student rejected from $p_{j'}$ according to $\mathcal{L}_k^{j'}$ (possibly $s_{i'} = s_i$). Moreover $s_{i'}$ was provisionally assigned to $p_{j'}$ during a previous \texttt{repeat-until} loop iteration but $(s_{i'}, p_{j'}) \notin M$ in the current \texttt{repeat-until} loop iteration. Thus $(s_{i'}, p_{j'})$ has been deleted before the deletion of $(s_i, p_j)$ occurred; and thus, $(s_{i'}, p_{j'}) \notin M^*$, since $(s_i, p_j)$ is the first super-stable pair to be deleted. Further,  $l_k$ either prefers $s_{i'}$ to $s_i$ or is indifferent between them, since $s_i$ plays the role of $s_t$ at some for loop iteration in line 32.

We remark that no student who is provisionally assigned to some project in $M$ can be assigned to a project better than her current assignment in any super-stable matching. For otherwise, this would mean a super-stable pair must have been deleted before $(s_i, p_j)$, since each student who is assigned in $M$ applies to projects in the head of her preference list. So, either $s_{i'}$ is unassigned in $M^*$ or $s_{i'}$ prefers $p_{j'}$ to $M^*(s_{i'})$ or $s_i$ is indifferent between them. By the super-stability of $M^*$, $p_{j'}$ is full in $M^*$ and $l_k$ prefers every student in $M^*(p_{j'})$ to $s_{i'}$; for otherwise, $(s_{i'}, p_{j'})$ blocks $M^*$, a contradiction.

Let $l_{z_0} = l_k$, $p_{t_0} = p_{j'}$ and $s_{q_0} = s_{i'}$. Just before the deletion of $(s_i, p_j)$ occurred, $p_{t_0}$ is undersubscribed in $M$. Since $p_{t_0}$ is full in $M^*$, there exists some student $s_{q_1} \in M^*(p_{t_0}) \setminus M(p_{t_0})$. We note that $l_{z_0}$ prefers $s_{q_1}$ to $s_{q_0}$; for otherwise, $(s_{i'}, p_{j'})$ blocks $M^*$, a contradiction. Let $p_{t_1} = p_{t_0}$. Since $(s_i, p_j)$ is the first super-stable pair to be deleted, $s_{q_1}$ is assigned in $M$ to a project $p_{t_2}$ such that $s_{q_1}$ prefers $p_{t_2}$ to $p_{t_1}$. For otherwise, as each student applies to projects at the head of her preference list, that would mean $(s_{q_1}, p_{t_1})$ must have been deleted before $(s_i, p_j)$, a contradiction. We note that $p_{t_2} \neq p_{t_1}$, since $(s_{q_1} , p_{t_2}) \in M$ and $(s_{q_1} , p_{t_1}) \notin M$. Let $l_{z_1}$ be the lecturer who offers $p_{t_2}$. By the super-stability of $M^*$, either (i) or (ii) holds as follows:
\begin{enumerate}[(i)]

\item $p_{t_2}$ is full in $M^*$ and $l_{z_1}$ prefers the worst student/s in $M^*(p_{t_2})$ to $s_{q_1}$;
\item $p_{t_2}$ is undersubscribed in $M^*$, $l_{z_1}$ is full in $M^*$, $s_{q_1} \notin M^*(l_{z_1})$ and $l_{z_1}$ prefer the worst student/s in $M^*(l_{z_1})$ to $s_{q_1}$.
\end{enumerate}
Otherwise $(s_{q_1}, p_{t_2})$ blocks $M^*$. In case (i), there exists some student $s_{q_2} \in M^*(p_{t_2}) \setminus M(p_{t_2})$. Let $p_{t_3} = p_{t_2}$. In case (ii), there exists some student $s_{q_2} \in M^*(l_{z_1}) \setminus M(l_{z_1})$. We note that $l_{z_1}$ prefers $s_{q_2}$ to $s_{q_1}$. Now, suppose $M^*(s_{q_2}) = p_{t_3}$ (possibly $p_{t_3} = p_{t_2}$). It is clear that $s_{q_2} \neq s_{q_1}$. Applying similar reasoning as for $s_{q_1}$, $s_{q_2}$ is assigned in $M$ to a project $p_{t_4}$ such that $s_{q_2}$ prefers $p_{t_4}$ to $p_{t_3}$. Let $l_{z_2}$ be the lecturer who offers $p_{t_4}$. We are identifying a sequence $\langle s_{q_i}\rangle_{i \geq 1}$ of students, a sequence $\langle p_{t_i}\rangle_{i \geq 1}$ of projects, and a sequence $\langle l_{z_i}\rangle_{i \geq 1}$ of lecturers, such that, for each $i \geq 1$

\begin{enumerate}

\item $s_{q_{i}}$ prefers $p_{t_{2i}}$ to $p_{t_{2i-1}}$,
\item $(s_{q_i}, p_{t_{2i}}) \in M$ and $(s_{q_i}, p_{t_{2i - 1}}) \in M^*$,
\item $l_{z_i}$ prefers $s_{q_{i+1}}$ to $s_{q_{i}}$; also, $l_{z_i}$ offers both $p_{t_{2i}}$ and $p_{t_{2i+1}}$ (possibly $p_{t_{2i}} = p_{t_{2i+1}}$).
\end{enumerate}

First we claim that for each new project that we identify, $p_{t_{2i}} \neq p_{t_{2i-1}}$ for $i \geq 1$. Suppose $p_{t_{2i}} = p_{t_{2i-1}}$ for some $i \geq 1$. From above $s_{q_{i}}$ was identified by $l_{z_{i-1}}$ such that $(s_{q_{i}}, p_{t_{2i-1}}) \in M^* \setminus M$. Moreover $(s_{q_{i}}, p_{t_{2i}}) \in M$. Hence we reach a contradiction. Clearly, for each student $s_{q_i}$ that we identify, for $i \geq 1$ , $s_{q_i}$ must be assigned to distinct projects in $M$ and in $M^*$.

Next we claim that for each new student $s_{q_i}$ that we identify, $s_{q_i} \neq s_{q_t}$ for $1 \leq t < i$. We prove this by induction on $i$. For the base case, clearly $s_{q_2} \neq s_{q_1}$. We assume that the claim holds for some $i \geq 1$, i.e., the sequence $s_{q_{1}}, s_{q_2}, \ldots, s_{q_{i}}$ consists of distinct students. We show that the claim holds for $i+1$, i.e., the sequence $s_{q_{1}}, s_{q_2}, \ldots, s_{q_{i}}, s_{q_{i+1}}$ also consists of distinct students. Clearly $s_{q_{i+1}} \neq s_{q_{i}}$ since $l_{z_{i}}$ prefers $s_{q_{i+1}}$ to $s_{q_{i}}$. Thus, it suffices to show that $s_{q_{i+1}} \neq s_{q_{j}}$ for $1 \leq j \leq i-1$. Now, suppose $s_{q_{i+1}} = s_{q_{j}}$ for $1 \leq j \leq i-1$. This implies that $s_{q_{j}}$ was identified by $l_{z_{i}}$ and clearly $l_{z_{i}}$ prefers $s_{q_{j}}$ to $s_{q_{j-1}}$. Now since $s_{q_{i+1}}$ was also identified by $l_{z_{i}}$ to avoid the blocking pair $(s_{q_i}, p_{t_{2_i}})$ in $M^*$, it follows that either
(i) $p_{t_{2i}}$ is full in $M^*$, or
(ii) $p_{t_{2i}}$ is undersubscribed in $M^*$ and $l_{z_{i}}$ is full in $M^*$. We consider each cases further as follows.
\begin{enumerate}[(i)]

\item If $p_{t_{2i}}$ is full in $M^*$, we know that $(s_{q_{i}}, p_{t_{2i}}) \in M \setminus M^*$. Moreover $s_{q_j}$ was identified by $l_{z_{i+1}}$ because of case (i). Furthermore $(s_{q_{j-1}}, p_{t_{2i}}) \in M \setminus M^*$. In this case, $p_{t_{2i+1}} = p_{t_{2i}}$ and we have that
$$(s_{q_{i}}, p_{t_{2i+1}})\in M \setminus M^* \mbox{ and } (s_{q_{i+1}}, p_{t_{2i+1}}) \in M^* \setminus M,$$ 
$$(s_{q_{j-1}}, p_{t_{2i+1}}) \in M \setminus M^* \mbox{ and } (s_{q_{j}}, p_{t_{2i+1}}) \in M^* \setminus M.$$
By the inductive hypothesis, the sequence $s_{q_{1}}, s_{q_2}, \ldots, s_{q_{j-1}}, $ $s_{q_j}, \ldots, s_{q_{i}}$ consists of distinct students. This implies that $s_{q_{i}} \neq s_{q_{j-1}}$. Thus since $p_{t_{2i+1}}$ is full in $M^*$, $l_{z_{i}}$ should have been able to identify distinct students $s_{q_j}$ and $s_{q_{i+1}}$ to avoid the blocking pairs $(s_{q_{j-1}}, p_{t_{2i+1}})$ and $(s_{q_{i}}, p_{t_{2i+1}})$ respectively in $M^*$, a contradiction.
\item $p_{t_{2i}}$ is undersubscribed in $M^*$ and $l_{z_{i}}$ is full in $M^*$. Similarly as in case (i) above, we have that
$$s_{q_{i}} \in M(l_{z_i}) \setminus M^*(l_{z_i}) \mbox{ and } s_{q_{i+1}} \in M^*(l_{z_i}) \setminus M(l_{z_i}),$$ 
$$s_{q_{j-1}} \in M(l_{z_i}) \setminus M^*(l_{z_i}) \mbox{ and } s_{q_{j}} \in M^*(l_{z_i}) \setminus M(l_{z_i}).$$
Since $s_{q_{i}} \neq s_{q_{j-1}}$ and $l_{z_{i}}$ is full in $M^*$, $l_{z_{i}}$ should have been able to identify distinct students $s_{q_j}$ and $s_{q_{i+1}}$ corresponding to students $s_{q_{j-1}}$ and $s_{q_{i}}$ respectively, a contradiction.
\end{enumerate}  

This completes the induction step. As the sequence of distinct students and projects we are identifying is infinite, we reach an immediate contradiction. \qed \end{proof}

Lemmas \ref{lemma:super-pair-deletion-within} and \ref{lemma:super-pair-deletion-outside} immediately give rise to Lemma \ref{pair-deletion}. The next lemma will be used as a tool in the proof of the remaining lemmas.

\begin{restatable}[]{lemma}{lecturerundersubscribedtool}

Let $M$ be the assignment at the termination of Algorithm {\sf SPA-ST-super} and let $M^*$ be any super-stable matching in $I$. Let $l_k$ be an arbitrary lecturer: (i) if $l_k$ is undersubscribed in $M^*$ then every student who is assigned to $l_k$ in $M$ is also assigned to $l_k$ in $M^*$; and (ii) if $l_k$ is undersubscribed in $M$ then $l_k$ has the same number of assignees in $M^*$ as in $M$. 
\label{lemma:super-lecturer-undersubscribed-tool}
\end{restatable}
\begin{proof}
Let $l_k$ be an arbitrary lecturer. First, we show that (i) holds. Suppose otherwise, then there exists a student, say $s_i$, such that $s_i \in M(l_k) \setminus M^*(l_k)$. Moreover, there exists some project $p_j \in P_k$ such that $s_i \in M(p_j) \setminus M^*(p_j)$. By Lemma \ref{pair-deletion}, $s_i$ cannot be assigned to a project that she prefers to $p_j$ in $M^*$. Also, by the super-stability of $M^*$, $p_j$ is full in $M^*$ and $l_k$ prefers the worst student/s in $M^*(p_j)$ to $s_i$. 

Let $l_{z_0} = l_k$, $p_{t_0} = p_{j}$, and $s_{q_0} = s_{i}$.
As $p_{t_0}$ is full in $M^*$ and no project is oversubscribed in $M$, there exists some student $s_{q_1} \in M^*(p_{t_0}) \setminus M(p_{t_0})$ such that $l_{z_0}$ prefers $s_{q_1}$ to $s_{q_0}$. Let $p_{t_1} = p_{t_0}$. By Lemma \ref{pair-deletion}, $s_{q_1}$ is assigned in $M$ to a project $p_{t_2}$ such that $s_{q_1}$ prefers $p_{t_2}$ to $p_{t_1}$. 
We note that $s_{q_1}$ cannot be indifferent between $p_{t_2}$ and $p_{t_1}$; for otherwise, as each student applies to projects at the head of her preference list, since $(s_{q_1}, p_{t_1}) \notin M$, that would mean $(s_{q_1}, p_{t_1})$ must have been deleted during the algorithm's execution, contradicting Lemma \ref{pair-deletion}. It follows that $s_{q_1} \in M(p_{t_2}) \setminus M^*(p_{t_2})$. Let $l_{z_1}$ be the lecturer who offers $p_{t_2}$. By the super-stability of $M^*$, either (i) or (ii) holds as follows:

\begin{enumerate}[(i)]

\item $p_{t_2}$ is full in $M^*$ and $l_{z_1}$ prefers the worst student/s in $M^*(p_{t_2})$ to $s_{q_1}$;
\item $p_{t_2}$ is undersubscribed in $M^*$, $l_{z_1}$ is full in $M^*$, $s_{q_1} \notin M^*(l_{z_1})$ and $l_{z_1}$ prefers the worst student/s in $M^*(l_{z_1})$ to $s_{q_1}$.
\end{enumerate}

Otherwise $(s_{q_1}, p_{t_2})$ blocks $M^*$. In case (i), there exists some student $s_{q_2} \in M^*(p_{t_2}) \setminus M(p_{t_2})$. Let $p_{t_3} = p_{t_2}$. In case (ii), there exists some student $s_{q_2} \in M^*(l_{z_1}) \setminus M(l_{z_1})$. We note that $l_{z_1}$ prefers $s_{q_2}$ to $s_{q_1}$. Now, suppose $M^*(s_{q_2}) = p_{t_3}$ (possibly $p_{t_3} = p_{t_2}$). It is clear that $s_{q_2} \neq s_{q_1}$. Applying similar reasoning as for $s_{q_1}$, student $s_{q_2}$ is assigned in $M$ to a project $p_{t_4}$ such that $s_{q_2}$ prefers $p_{t_4}$ to $p_{t_3}$. Let $l_{z_2}$ be the lecturer who offers $p_{t_4}$. We are identifying a sequence $\langle s_{q_i}\rangle_{i \geq 1}$ of students, a sequence $\langle p_{t_i}\rangle_{i \geq 1}$ of projects, and a sequence $\langle l_{z_i}\rangle_{i \geq 1}$ of lecturers, such that, for each $i \geq 1$

\begin{enumerate}

\item $s_{q_{i}}$ prefers $p_{t_{2i}}$ to $p_{t_{2i-1}}$,
\item $(s_{q_i}, p_{t_{2i}}) \in M$ and $(s_{q_i}, p_{t_{2i - 1}}) \in M^*$,
\item $l_{z_i}$ prefers $s_{q_{i+1}}$ to $s_{q_{i}}$; also, $l_{z_i}$ offers both $p_{t_{2i}}$ and $p_{t_{2i+1}}$ (possibly $p_{t_{2i}} = p_{t_{2i+1}}$).
\end{enumerate}

Following a similar argument as in the proof of Lemma~\ref{lemma:super-pair-deletion-outside}, we can identify an infinite sequence of distinct students and projects, a contradiction. Hence, if $l_k$ is undersubscribed in $M^*$ then every student who is assigned to $l_k$ in $M$ is also assigned to $l_k$ in $M^*$. 

Next, we show that (ii) holds. By the first claim, any lecturer who is full in $M$ is also full in $M^*$, and any lecturer who is undersubscribed in $M$ has as many assignees in $M^*$ as she has in $M$. Hence 
\begin{eqnarray}
\label{ineq:undersubscribed-lecturer-1}
\sum_{l_k \in \mathcal{L}}{|M(l_k)|} \leq \sum_{l_k \in \mathcal{L}}{|M^*(l_k)|} \enspace.
\end{eqnarray}
We note that if a student $s_{i}$ is unassigned in $M$, by Lemma \ref{pair-deletion}, $s_{i}$ is unassigned in $M^*$. Equivalently, if $s_{i}$ is assigned in $M^*$ then $s_{i}$ is assigned in $M$. Let $S_1$ denote the set of students who are assigned to at least one project in $M$, and let $S_2$ denote the set of students who are assigned to a project in $M^*$; it follows that $|S_2| \leq |S_1|$. Further, we have that
\begin{eqnarray}
\label{ineq:undersubscribed-lecturer-2}
\sum_{l_k \in \mathcal{L}}{|M^*(l_k)|} = |S_2| \leq |S_1| \leq \sum_{l_k \in \mathcal{L}}{|M(l_k)|},
\end{eqnarray}
From Inequalities \eqref{ineq:undersubscribed-lecturer-1} and \eqref{ineq:undersubscribed-lecturer-2}, it follows that $|M(l_k)| = |M^*(l_k)|$ for each $l_k \in \mathcal{L}$.
\qed  \end{proof}
The next three lemmas deal with the case that Algorithm {\sf SPA-ST-super} reports the non-existence of a super-stable matching in $I$.
\begin{restatable}[]{lemma}{studentlemma}
\label{lemma-super-multi-assignment}
If a student is assigned to two or more projects at the termination of Algorithm {\sf SPA-ST-super} then $I$ admits no super-stable matching.
\end{restatable}
\begin{proof}
Let $M$ be the assignment at the termination of the algorithm. Suppose for a contradiction that there exists a super-stable matching $M^*$ in $I$. Suppose that a student is assigned to two or more projects in $M$. Then either (a) any two of these projects are offered by different lecturers or (b) all of these projects are offered by the same lecturer. 

Firstly, suppose (a) holds. Then some lecturer has fewer assignees in $M^*$ than in $M$. Suppose not, then
\begin{eqnarray}
\label{eqn-multiple-assignment-1}
\sum_{l_k \in \mathcal{L}}{|M^*(l_k)|} \geq \sum_{l_k \in  \mathcal{L}}{|M(l_k)|}\enspace.
\end{eqnarray}
Let $S_1$ and $S_2$ be as defined in the proof of Lemma \ref{lemma:super-lecturer-undersubscribed-tool}, it follows that $|S_2| \leq |S_1|$. Hence,
\begin{eqnarray}
\label{eqn-multiple-assignment-2}
\sum_{l_k \in  \mathcal{L}}{|M^*(l_k)|} = |S_2| \leq |S_1| < \sum_{l_k \in  \mathcal{L}}{|M(l_k)|},
\end{eqnarray}
since some student in $S_1$ is assigned in $M$ to two or more projects offered by different lecturers. Inequality \eqref{eqn-multiple-assignment-2} contradicts Inequality \eqref{eqn-multiple-assignment-1}. Hence, our claim is established. As some lecturer $l_k$ has fewer assignees in $M^*$ than in $M$, it follows that $l_k$ is undersubscribed in $M^*$, since no lecturer is oversubscribed in $M$. In particular, there exists some project $p_j \in P_k$ and some student, say $s_i$, such that $p_j$ is undersubscribed in $M^*$ and $(s_i, p_j) \in M \setminus M^*$. Since $(s_i, p_j) \in M$, then $p_j$ must have been in the head of $s_i$'s preference list when $s_i$ applied to $p_j$ during the algorithm's execution. By Lemma \ref{pair-deletion}, either $s_i$ is unassigned in $M^*$ or $s_i$ prefers $p_j$ to $M^*(s_i)$ or $s_i$ is indifferent between them. Hence $(s_i, p_j)$ blocks $M^*$, a contradiction. 

Next, suppose (b) holds. Then $|S_1| \leq \sum_{l_k \in \mathcal{L}} |M(l_k)|$. As in case (a), since $|S_2|\leq |S_1|$, it follows that
$$\sum_{l_k \in \mathcal{L}} |M^*(l_k)| \leq \sum_{l_k \in \mathcal{L}} |M(l_k)|\enspace.$$

Suppose first that $|M^*(l_k)| < |M(l_k)|$ for some $l_k \in \mathcal{L}$. Then $l_k$ has fewer assignees in $M^*$ than in $M$, and following a similar argument as in case (a) above, we reach an immediate contradiction. Hence, $|M^*(l_k)| = |M(l_k)|$ for all $l_k \in \mathcal{L}$. For each $l_k \in \mathcal{L}$, we claim that every student who is assigned to $l_k$ in $M$ is also assigned to $l_k$ in $M^*$. Suppose otherwise. Let $l_{z_1}$ be an arbitrary lecturer in $\mathcal{L}$. Then there exists some student $s_{q_1} \in M(l_{z_1}) \setminus M^*(l_{z_1})$. Let $M(s_{q_1}) = p_{t_2}$. By Lemma \ref{pair-deletion}, $s_{q_1}$ is assigned in $M^*$ to a project $p_{t_1}$ such that $s_{q_1}$ prefers $p_{t_2}$ to $p_{t_1}$. Clearly, $p_{t_1}$ is not offered by $l_{z_1}$, since $s_{q_1} \in M(l_{z_1}) \setminus M^*(l_{z_1})$. We also note that $s_{q_1}$ cannot be indifferent between $p_{t_2}$ and $p_{t_1}$. Otherwise, the argument follows from (a), since $s_{q_1}$ is assigned in $M$ to two projects offered by different lecturers, and we reach an immediate contradiction. By the super-stability of $M^*$, either (i) or (ii) holds as follows:
\begin{enumerate}[(a)]

\item $p_{t_2}$ is full in $M^*$ and $l_{z_1}$ prefers every student in $M^*(p_{t_2})$ to $s_{q_1}$;
\item $p_{t_2}$ is undersubscribed in $M^*$, $l_{z_1}$ is full in $M^*$ and $l_{z_1}$ prefers every student in $M^*(l_{z_1})$ to $s_{q_1}$.
\end{enumerate}

Otherwise, $(s_{q_1}, p_{t_2})$ blocks $M^*$. In case (i), there exists some student $s_{q_2} \in M^*(p_{t_2}) \setminus M(p_{t_2})$. Let $p_{t_3} = p_{t_2}$. In case (ii), there exists some student $s_{q_2} \in M^*(l_{z_1}) \setminus M(l_{z_1})$. We note that $l_{z_1}$ prefers $s_{q_2}$ to $s_{q_1}$, and clearly $s_{q_2} \neq s_{q_1}$. Let $M^*(s_{q_2}) = p_{t_3}$ (possibly $p_{t_3} = p_{t_2}$). Applying similar reasoning as for $s_{q_1}$, student $s_{q_2}$ is assigned in $M$ to a project $p_{t_4}$ such that $s_{q_2}$ prefers $p_{t_4}$ to $p_{t_3}$. We are identifying a sequence $\langle s_{q_i}\rangle_{i \geq 1}$ of students, a sequence $\langle p_{t_i}\rangle_{i \geq 1}$ of projects, and a sequence $\langle l_{z_i}\rangle_{i \geq 1}$ of lecturers, such that, for each $i \geq 1$

\begin{enumerate}

\item $s_{q_{i}}$ prefers $p_{t_{2i}}$ to $p_{t_{2i-1}}$,
\item $(s_{q_i}, p_{t_{2i}}) \in M$ and $(s_{q_i}, p_{t_{2i - 1}}) \in M^*$,
\item $l_{z_i}$ prefers $s_{q_{i+1}}$ to $s_{q_{i}}$; also, $l_{z_i}$ offers both $p_{t_{2i}}$ and $p_{t_{2i+1}}$ (possibly $p_{t_{2i}} = p_{t_{2i+1}}$).
\end{enumerate}

Following a similar argument as in the proof of Lemma~\ref{lemma:super-pair-deletion-outside}, we can identify an infinite sequence of distinct students and projects, a contradiction.

Now, let $s_i$ be an arbitrary student such that $s_i$ is assigned in $M$ to two or more projects offered by a lecturer, say $l_k$. Then $s_i \in M^*(l_k)$. Moreover, there exists some project $p_j \in P_k$ such that $(s_i, p_j) \in M \setminus M^*$. We claim that $p_j$ is undersubscribed in $M^*$. Suppose otherwise. Let $l_{z_0} = l_k$, $p_{t_0} = p_j$ and $s_{q_0} = s_i$. Then there exists some student $s_{q_1} \in M^*(p_{t_0}) \setminus M(p_{t_0})$, since $p_{t_0}$ is not oversubscribed in $M$ and $s_{q_0} \in M(p_{t_0}) \setminus M^*(p_{t_0})$. Again, by Lemma \ref{pair-deletion}, $s_{q_1}$ is assigned in $M$ to a project $p_{t_1}$ such that $s_{q_1}$ prefers $p_{t_1}$ to $p_{t_0}$. Let $l_{z_1}$ be the lecturer who offers $p_{t_1}$. Following a similar argument as in the proof of Lemma~\ref{lemma:super-pair-deletion-outside}, we can identify a sequence of distinct students and projects, and as this sequence is infinite, we reach a contradiction. Hence our claim holds, i.e., $p_j$ is undersubscribed in $M^*$. Finally, since $s_i$ cannot be assigned to any project that she prefers to $p_j$ in $M^*$ and since $(s_i, p_j) \in M^*(l_k)$, we have that $(s_i, p_j)$ blocks $M^*$, a contradiction.
\qed  \end{proof}
\begin{restatable}[]{lemma}{lecturerlemma}
\label{lemma:super-lec-full-under}
If some lecturer $l_k$ becomes full during some execution of Algorithm {\sf SPA-ST-super} and $l_k$ subsequently ends up undersubscribed at the termination of the algorithm, then $I$ admits no super-stable matching.
\end{restatable}
\begin{proof}
Let $M$ be the assignment at the termination of the algorithm. Suppose for a contradiction that there exists a super-stable matching $M^*$ in $I$. 
Let $l_k$ be the lecturer who became full during some execution of the algorithm and subsequently ends up undersubscribed in $M$. By Lemma \ref{lemma:super-lecturer-undersubscribed-tool}, $|M(l_k)| = |M^*(l_k)|$ and thus $l_k$ is undersubscribed in $M^*$. At the point in the algorithm where $l_k$ became full (line 22), we note that none of the projects offered by $l_k$ is oversubscribed. Since $l_k$ ended up undersubscribed in $M$, it follows that there is some project $p_j \in P_k$ that has fewer assignees in $M$ at the termination of the algorithm than it had at some point during the algorithm's execution, thus $p_j$ is undersubscribed in $M$.

We claim that each project offered by $l_k$ has the same number of assignees in $M^*$ as in $M$. Suppose otherwise, then there is some project $p_t \in P_k$ such that $|M^*(p_t)| < |M(p_t)|$; thus $p_t$ is undersubscribed in $M^*$, since no project is oversubscribed in $M$. It follows that there exists some student $s_r \in M(p_t) \setminus M^*(p_t)$. By Lemma \ref{pair-deletion}, $s_r$ is either unassigned in $M^*$ or prefers $p_t$ to $M^*(s_r)$. Since $l_k$ is undersubscribed in $M^*$, $(s_r, p_t)$ blocks $M^*$, a contradiction. Hence $|M^*(p_t)| \geq |M(p_t)|$. Moreover, since $|M(l_k)| = |M^*(l_k)|$, we have that $|M(p_t)| = |M^*(p_t)|$ for all $p_t \in P_k$.

Hence $p_j$ undersubscribed in $M$ implies that $p_j$ is undersubscribed in $M^*$. Moreover, there is some student $s_i$ who was provisionally assigned to $p_j$ at some point during the execution of the algorithm but $s_i$ is not assigned to $p_j$ in $M$. Thus, the pair $(s_i, p_j)$ was deleted during the algorithm's execution, so that $(s_i, p_j) \notin M^*$ by Lemma \ref{pair-deletion}. It follows that either $s_i$ is unassigned in $M^*$ or $s_i$ prefers $p_j$ to $M^*(s_i)$ or $s_i$ is indifferent between them. Hence, $(s_i, p_j)$ blocks $M^*$, a contradiction. 
\qed  \end{proof}
\begin{restatable}[]{lemma}{projectlemma}
\label{lemma-super-proj-full-under}
If the pair $(s_i, p_j)$ was deleted during some execution of Algorithm {\sf SPA-ST-super}, and at the termination of the algorithm $s_i$ is not assigned to a project better than $p_j$, and each of $p_j$ and $l_k$ is undersubscribed, then $I$ admits no super-stable matching.
\end{restatable}
\begin{proof}
Suppose for a contradiction that there exists a super-stable matching $M^*$ in $I$. Let $(s_i, p_j)$ be a pair that was deleted during an arbitrary execution $E$ of the algorithm. This implies that $(s_i, p_j) \notin M^*$ by Lemma \ref{pair-deletion}. Let $M$ be the assignment at the termination of $E$. By the hypothesis of the lemma, $l_k$ is undersubscribed in $M$. This implies that $l_k$ is undersubscribed in $M^*$, by Lemma~\ref{lemma:super-lecturer-undersubscribed-tool}. Since $p_j$ is offered by $l_k$, and $p_j$ is undersubscribed in $M$, it follows from the proof of Lemma \ref{lemma:super-lec-full-under} that $p_j$ is undersubscribed in $M^*$. Further, by the hypothesis of the lemma, either $s_i$ is unassigned in $M$, or $s_i$ prefers $p_j$ to $M(s_i)$ or is indifferent between them. By Lemma \ref{pair-deletion}, this is true for $s_i$ in $M^*$. Hence $(s_i, p_j)$ blocks $M^*$, a contradiction.
\qed \end{proof}
The next lemma shows that the final assignment may be used to determine the existence, or otherwise, of a super-stable matching in $I$.
\begin{restatable}[]{lemma}{nosuperstablematchinglemma}
\label{lemma-super-correctness}
If at the termination of Algorithm {\sf SPA-ST-super}, the assignment $M$ is not super-stable in $I$ then no super-stable matching exists in $I$.
\end{restatable}
\begin{proof}
Suppose $M$ is not super-stable in $I$. If some student $s_i$ is assigned to two or more projects in $M$ then $I$ admits no super-stable matching, by Lemma \ref{lemma-super-multi-assignment}. Hence every student is assigned to at most one project in $M$. Moreover, since no project or lecturer is oversubscribed in $M$, it follows that $M$ is a matching. Let $(s_i, p_j)$ be a blocking pair for $M$, then $s_i$ is either unassigned in $M$ or prefers $p_j$ to $M(s_i)$ or is indifferent between them. Whichever is the case, $(s_i, p_j)$ has been deleted. Let $l_k$ be the lecturer who offers $p_j$. In what follows, we will identify the point in the algorithm at which $(s_i, p_j)$ was deleted, and consequently, we will arrive at a conclusion that no super-stable matching exists.

Firstly, suppose $(s_i, p_j)$ was deleted as a result of $p_j$ being full or oversubscribed (on lines 12 or 21). Suppose $p_j$ is full in $M$. Then $(s_i, p_j)$ cannot block $M$ irrespective of whether $l_k$ is undersubscribed or full in $M$, since $l_k$ prefers the worst assigned student/s in $M(p_j)$ to $s_i$. Hence $p_j$ is undersubscribed in $M$. As $p_j$ was previously full, each pair $(s_t, p_u)$, for each $s_t$ that is no better than $s_i$ at the tail of $\mathcal{L}_k$ and each $p_u \in P_k \cap A_t$, would have been deleted on line 34 of the algorithm. Thus, if $l_k$ is full in $M$ then $(s_i, p_j)$ does not block $M$. Suppose $l_k$ is undersubscribed in $M$. If $l_k$ was full at some point during the execution of the algorithm then $I$ admits no super-stable matching, by Lemma \ref{lemma:super-lec-full-under}. Hence $l_k$ was never full during the algorithm's execution. Recall that each of $p_j$ and $l_k$ is undersubscribed in $M$. As $(s_i, p_j)$ is a blocking pair of $M$, $s_i$ cannot be assigned in $M$ to a project that she prefers to $p_j$. Hence $I$ admits no super-stable matching, by Lemma \ref{lemma-super-proj-full-under}. 

Next, suppose $(s_i, p_j)$ was deleted as a result of $l_k$ being full or oversubscribed (on lines 16 or 26), $(s_i, p_j)$ could only block $M$ if $l_k$ is undersubscribed in $M$. If this is the case then $I$ admits no super-stable matching, by Lemma \ref{lemma:super-lec-full-under}.

Finally, suppose $(s_i, p_j)$ was deleted (on line 34) because some other project $p_{j'}$ offered by $l_k$ was previously full and ended up undersubscribed on line 28. Then $l_k$ must have identified the most-preferred student, say $s_r$, who was previously assigned to $p_{j'}$ but subsequently got rejected from $p_{j'}$. At this point, $s_i$ is at the tail of $\mathcal{L}_k$ and $s_i$ is no better than $s_r$ in $\mathcal{L}_k$. Moreover, every project offered by $l_k$ that $s_i$ finds acceptable would have been deleted from $s_i$'s preference list at the for loop iteration in line 34. If $p_j$ is full in $M$ then $(s_i,p_j)$ does not block $M$. Hence $p_j$ is undersubscribed in $M$. If $l_k$ is full in $M$ then $(s_i, p_j)$ does not block $M$, since $s_i \notin M(l_k)$ and $l_k$ prefers the worst student/s in $M(l_k)$ to $s_i$. Hence $l_k$ is undersubscribed in $M$. Again by Lemma \ref{lemma-super-proj-full-under}, $I$ admits no super-stable matching.

Since $(s_i, p_j)$ is an arbitrary pair, this implies that $I$ admits no super-stable matching.
\qed  \end{proof}

The next lemma shows that Algorithm {\sf SPA-ST-super} may be implemented to run in linear time.
\begin{restatable}[]{lemma}{lineartimelemma}
\label{lemma-super-complexity}
Algorithm {\sf SPA-ST-super} may be implemented to run in $O(L)$ time and $O(n_1n_2)$ space, where $n_1$, $n_2$, and $L$ are the number of students, number of projects, and the total length of the preference lists, respectively, in $I$.
\end{restatable}
\begin{proof}
The algorithm's time complexity depends on how efficiently we can execute the operation of a student applying to a project and the operation of deleting a (student, project) pair, each of which occur once for any (student, project) pair. It turns out that both operations can be implemented to run in constant time, giving Algorithm {\sf SPA-ST-super} an overall complexity of $\Theta(L)$, where $L$ is the total length of all the preference lists. In what follows, we describe the non-trivial aspects of such an implementation. We remark that the data structures discussed here are inspired by, and extend, those detailed in \cite[Section 3.3]{AIM07} for Algorithm {\sf SPA}-student.

For each student $s_i$, build an array $\mathit{position}_{s_i}$, where $\mathit{position}_{s_i}(p_j)$ is the position of project $p_j$ in $s_i$'s preference list. For example, if $s_i$'s preference list is $(p_2 \; p_5 \; p_3) \; p_7 \; (p_6 \; p_1)$ then $\mathit{position}_{s_i}(p_5) = 2$ and $\mathit{position}_{s_i}(p_1) = 6$. In general, position captures the order in which the projects appear in the preference list when read from left to right, ignoring any ties. Represent $s_i$'s preference list by embedding doubly linked lists in an array $\mathit{preference}_{s_i}$. For each project $p_j \in A_i$, $\mathit{preference}_{s_i}(\mathit{position}_{s_i}(p_j))$ stores the list node containing $p_j$. This node contains two next pointers (and two previous pointers) --  one to the next project in $s_i$'s preference list (after deletions, this project may not be located at the next array position), and another pointer to the next project $p_{j'}$ in $s_i$'s preference list, where $p_{j'}$ and $p_j$ are both offered by the same lecturer. Construct the latter list by traversing through $s_i$'s preference list, using a temporary array to record the last project in the list offered by each lecturer. Use virtual initialisation (described in \cite[p.~149]{BB96}) for these arrays, since the overall $O(n_1 n_3)$ initialisation may be too expensive.

To represent the ties in $s_i$'s preference list, build an array $\mathit{successor}_{s_i}$.  For each project $p_j$ in $s_i$'s preference list, $\mathit{successor}_{s_i}(\mathit{position}_{s_i}(p_j))$ stores the \texttt{true} boolean if $p_j$ is tied with its successor
in $A_i$  and \texttt{false} otherwise. After the deletion of any (student, project) pair, update the successor booleans. As an illustration, with respect to $s_i$'s preference list given in the previous paragraph, $\mathit{successor}_{s_i}$ is the array [\texttt{true, true, false, false, true, false}]. 
Now, suppose $p_3$ was deleted from $s_i$'s preference list, since $\mathit{successor}_{s_i}(\mathit{position}_{s_i}(p_3))$ is \texttt{false} and $\mathit{successor}_{s_i}(\mathit{position}_{s_i}(p_5))$ is \texttt{true}, set $\mathit{successor}_{s_i}$ $(\mathit{position}_{s_i}(p_5))$ to \texttt{false} (since $p_5$ is the predecessor of $p_3$). Clearly using these data structures, we can find the next project at the head of each student's preference list, find the next project offered by a given lecturer on each student's preference list, as well as delete a project from a given student's preference list in constant time.

For each lecturer $l_k$, build two arrays $\mathit{preference}_{l_k}$ and $\mathit{successor}_{l_k}$, where $\mathit{preference}_{l_k}(s_i)$ is the position of student $s_i$ in $l_k$'s preference list, and $\mathit{successor}_{l_k}$ $(\mathit{preference}_{l_k}(s_i))$ stores the position of the first strict successor (with respect to position) of $s_i$ in $\mathcal{L}_k$ or a null value if $s_i$ has no strict successor\footnote{For example, if $l_k$'s preference list is $s_5 \; (s_3 \; s_1 \; s_6) \; s_7 \; (s_2 \; s_8)$ then $\mathit{successor}_{l_k}$ is the array $[2 \; 5 \; 5 \; 5 \; 6 \; 0 \; 0]$.}. Represent $l_k$'s preference list (i.e., $\mathcal{L}_k$) by the array $\mathit{preference}_{l_k}$, with an additional pointer, $\mathit{last}_{l_k}$. Initially, $\mathit{last}_{l_k}$ stores the index of the last position in $\mathit{preference}_{l_k}$. To represent the ties in $l_k$'s preference list, build an array $\mathit{predecessor}_{l_k}$. For each $s_i \in \mathcal{L}_k$, $\mathit{predecessor}_{l_k}(\mathit{preference}_{l_k}(s_i))$ stores the \texttt{true} boolean if $s_i$ is tied with its predecessor in $\mathcal{L}_k$ and \texttt{false} otherwise. 

When $l_k$ becomes full, make $\mathit{last}_{l_k}$ equivalent to $l_k$'s worst assigned student through the following method. Perform a backward traversal through the array $\mathit{preference}_{l_k}$, starting at $\mathit{last}_{l_k}$, and continuing until $l_k$'s worst assigned student, say $s_{i'}$, is encountered (each student stores a pointer to their assigned project, or a special null value if unassigned). Deletions must be carried out in the preference list of each student who is worse than $s_{i'}$ on $l_k$'s preference list (precisely those students whose position in $\mathit{preference}_{l_k}$ is greater than or equal to that stored in $\mathit{successor}_{l_k}(\mathit{preference}_{l_k}(s_{i'}))$)\footnote{For efficiency, we remark that it is not necessary to make deletions from the preference lists of lecturers or projected preference lists of lecturers for each project the lecturer offers, since the while loop of Algorithm {\sf SPA-ST-super} involves students applying to projects in the head of their preference list.}. 

When $l_k$ becomes oversubscribed, we can find and delete the students at the tail of $l_k$ by performing a backward traversal through the array $\mathit{preference}_{l_k}$, starting at $\mathit{last}_{l_k}$, and continuing until we encounter a student, say $s_{i'}$, such that $\mathit{predecessor}_{l_k}(\mathit{preference}_{l_k}(s_{i'}))$ stores the \texttt{false} boolean. If $l_k$ becomes undersubscribed after we break the assignment of students encountered on this traversal (including $s_{i'}$) to $l_k$, rather than update $\it{last}_{l_k}$ immediately, which could be expensive, we wait until $l_k$ becomes full again. The cost of these traversals taken over the algorithm's execution is thus linear in the length of $l_k$'s preference list.

For each project $p_j$ offered by $l_k$, build the arrays $\mathit{preference}_{p_j}$, $\mathit{successor}_{p_j}$ and  $\mathit{predecessor}_{p_j}$ corresponding to $\mathcal{L}_k^j$,  as described in the previous paragraph for $\mathcal{L}_k$. Represent the projected preference list of $l_k$ for $p_j$ (i.e., $\mathcal{L}_k^j$) by the array $\mathit{preference}_{p_j}$, with an additional pointer, $\mathit{last}_{p_j}$. These project preference arrays are used in much the same way as the lecturer preference arrays

Since we only visit a student at most twice during these backward traversals, once for the lecturer and once for the project, the asymptotic running time remains linear.
 \qed \end{proof}

\noindent 
Lemma \ref{pair-deletion} shows that there is an optimality property for each assigned student in any super-stable matching found by the algorithm, whilst Lemma \ref{lemma-super-correctness} establishes the correctness of Algorithm {\sf SPA-ST-super}. The following theorem collects together Lemmas \ref{pair-deletion}, \ref{lemma-super-correctness} and \ref{lemma-super-complexity}.
\begin{theorem}
\label{thrm:super-optimality}
For a given instance $I$ of {\sc spa-st}, Algorithm {\sf SPA-ST-super} determines, in $O(L)$ time and $O(n_1n_2)$ space, whether or not a super-stable matching exists in $I$. If such a matching does exist, all possible executions of the algorithm find one in which each assigned student is assigned to the best project that she could obtain in any super-stable matching, and each unassigned student is unassigned in all super-stable matchings.
\end{theorem}

Given the optimality property established by Theorem \ref{thrm:super-optimality}, we define the super-stable matching found by Algorithm {\sf SPA-ST-super} to be \textit{student-optimal}. 

\subsection{Properties of super-stable matchings in {\sc spa-st}}
\label{subsect:properties}
In this section, we consider properties of the set of super-stable matchings in an instance of {\sc spa-st}. We show that the Unpopular Projects Theorem for {\sc spa-s} (see Theorem \ref{thrm:rural-spa-s}) holds for {\sc spa-st} under super-stability.
\begin{restatable}[]{theorem}{upt}
\label{thrm:upt}
For a given instance $I$ of {\sc spa-st}, the following holds:
\begin{enumerate}
\item each lecturer is assigned the same number of students in all super-stable matchings;
\item exactly the same students are unassigned in all super-stable matchings;
\item a project offered by an undersubscribed lecturer has the same number of students in all super-stable matchings.
\end{enumerate}
\end{restatable}

\begin{proof}
Let $M$ and $M^*$ be two arbitrary super-stable matchings in $I$. Let $I'$ be an instance of {\sc spa-s} obtained from $I$ by breaking the ties in $I$ in some way. Then by Proposition \ref{proposition1}, each of $M$ and $M^*$ is stable in $I'$.  Thus by Theorem \ref{thrm:rural-spa-s}, each lecturer is assigned the same number of students in $M$ and $M^*$, exactly the same students are unassigned in $M$ and $M^*$, and a project offered by an undersubscribed lecturer has the same number of students in $M$ and $M^*$.
\qed\end{proof}

\begin{figure}[t]
\centering
\renewcommand{\arraystretch}{1}
\begin{tabular}{llll}
\hline
Students' preferences & \qquad \qquad & Lecturers' preferences & offers \\ 
$s_1$: \;  $p_1$  &  & $l_1$: \; $s_5$ \; $s_6$ \; $s_4$ \;($s_1$ \; $s_2$) \; $s_3$ & $p_1$, $p_2$\\ 
$s_2$: \;($p_1$ \; $p_3$)  &  & $l_2$: \;  $s_3$ \; $s_4$ \; $s_5$ \; $s_6$ \; $s_2$ & $p_3$, $p_4$\\ 
$s_3$: \; $p_2$ \; $p_3$ &  & &\\
$s_4$: \; $p_2$ \; $p_3$ &  & &\\
$s_5$: \; $p_3$ \; $p_2$ &  & Project capacities: $c_1 = c_4 = 1, \; c_2 = c_3 = 2$& \\
$s_6$: \; $p_2$ \; $p_4$ &  & Lecturer capacities: $d_1 = 2$, $d_2 = 3$&\\ 
\hline
\end{tabular}
\caption{\label{fig:spa-st-instance-2} \small Instance $I_2$ of {\sc spa-st}.}
\end{figure}

To illustrate this, consider the {\sc spa-st} instance $I_2$ given in Fig.~\ref{fig:spa-st-instance-2}, which admits the super-stable matchings $M_1 = \{(s_3, p_3), (s_4, p_2), (s_5, p_3),$ $(s_6, p_2)\}$ and $M_2 = \{(s_3, p_3), (s_4, p_3), $ $(s_5, p_2), (s_6, p_2)\}$. Each of $l_1$ and $l_2$ is assigned the same number of students in both $M_1$ and $M_2$, illustrating part (1) of Theorem \ref{thrm:upt}.  Also, each of $s_1$ and $s_2$ is unassigned in both $M_1$ and $M_2$, illustrating part (2) of Theorem \ref{thrm:upt}. Finally, $l_2$ is undersubscribed in both $M_1$ and $M_2$, and each of $p_3$ and $p_4$ has the same number of students in both $M_1$ and $M_2$, illustrating part (3) of Theorem \ref{thrm:upt}.
%


\section{Empirical Evaluation}
\label{emprical-results}
In this section, we evaluate an implementation of Algorithm {\sf SPA-ST-super}. We implemented our algorithm in Python\footnote{https://github.com/sofiatolaosebikan/spa-st-super}, and performed our experiments on a system with dual Intel Xeon CPU E5-2640 processors with 64GB of RAM, running Ubuntu 17.10. For our experiment, we were primarily concerned with the following question: how does the nature of the preference lists in a given {\sc spa-st} instance affect the existence of a super-stable matching? 
\subsection{Datasets}
When generating random datasets, there are clearly several parameters that can be varied, such as the number of students, projects and lecturers; the lengths of the students' preference lists as well as a measure of the density of ties present in the preference lists. We denote by $t_d$, the measure of the density of ties present in the preference lists. In each student's preference list, the tie density $t_{d_s} \; (0 \leq t_{d_s} \leq 1)$ is the probability that some project is tied to its successor.  The tie density $t_{d_l}$ in each lecturer's preference list is defined similarly. At $t_{d_s} = t_{d_l} = 1$, each preference list comprises a single tie while at $t_{d_s} = t_{d_l} = 0$, no tie would exist in the preference lists, thus reducing the problem to an instance of {\sc spa-s}. 
\subsection{Experimental Setup}
For each range of values for the aforementioned parameters, we randomly generated a set of {\sc spa-st} instances, involving $n_1$ students (which we will henceforth refer to as the size of the instance), $0.5n_1$ projects, $0.2n_1$ lecturers and $1.5n_1$ total project capacity which was randomly distributed amongst the projects. The capacity for each lecturer $l_k$ was chosen uniformly at random to lie between the highest capacity of the projects offered by $l_k$ and the sum of the capacities of the projects that $l_k$ offers.\footnote{We remark that the parameter space was chosen to ensure that projects could typically accommodate more than one student, that the total capacity of the projects exceeded the number of students, and that each lecturer typically offered multiple projects, without reflecting any specific real-world application.} In each set, we measured the proportion of instances that admit a super-stable matching.

It is worth mentioning that when we varied the tie density on both the students' and lecturers' preference lists between $0.1$ and $0.5$, super-stable matchings were very elusive, even with an instance size of $100$ students. Thus, for the purpose of our experiment, we decided to choose a low tie density. 

\subsubsection{Correctness testing}
To test the correctness of our algorithm's implementation, we implemented an Integer Programming (IP) model for super-stability in {\sc spa-st} (see Appendix \ref{appendixA}) using the Gurobi optimisation solver in Python. We randomly generated $10,000$ {\sc spa-st} instances, each consisting of $100$ students and a constant ratio of projects, lecturers, project capacities and lecturer capacities as described above. Also, each student's preference list was fixed at $10$, with a tie density of $0.1$. With this setup, we verified consistency between the outcomes of our implementation of Algorithm {\sf SPA-ST-super} and our implementation of the IP-based algorithm in terms of the existence or otherwise of a super-stable matching.


\subsubsection{Experiment 1}
In our first experiment, we examined how the length of the students' preference lists affects the existence of a super-stable matching. We increased the number of students $n_1$ while maintaining a constant ratio of projects, lecturers, project capacities and lecturer capacities as described above. For various values of $n_1 \; (100 \leq n_1 \leq 1000)$ in increments of $100$, we varied the length of each student's preference list for various values of $x$ ($5 \leq x \leq 50$) in increments of $5$; and with each of these parameters, we randomly generated $1000$ instances. For all the preference lists, we set $t_{d_s} = t_{d_l} = 0.005$ (on average, $1$ out of $5$ students has a single tie of length $2$ in their preference list, and this holds similarly for the lecturers). 

The result, which is displayed in Fig.~\ref{super-experiment1}, shows that as we varied the length of the preference list, there was no significant uplift in the number of instances that admitted a super-stable matching. In most cases, we observed that the proportion of instances that admit a super-stable matching is slightly higher when the preference list length is $50$ compared to when the preference list length is $5$. The result also shows that the proportion of instances that admit a super-stable matching decreases as the number of students increases. Further, we recorded the time taken for our algorithm's implementation to terminate, and as can be seen in Table \ref{fig:super-time-table}, for an instance size of $1000$ and preference list length $50$, the algorithm terminates in approximately $0.4$ second.

\begin{figure}[H]
\centering
\includegraphics[scale=0.37]{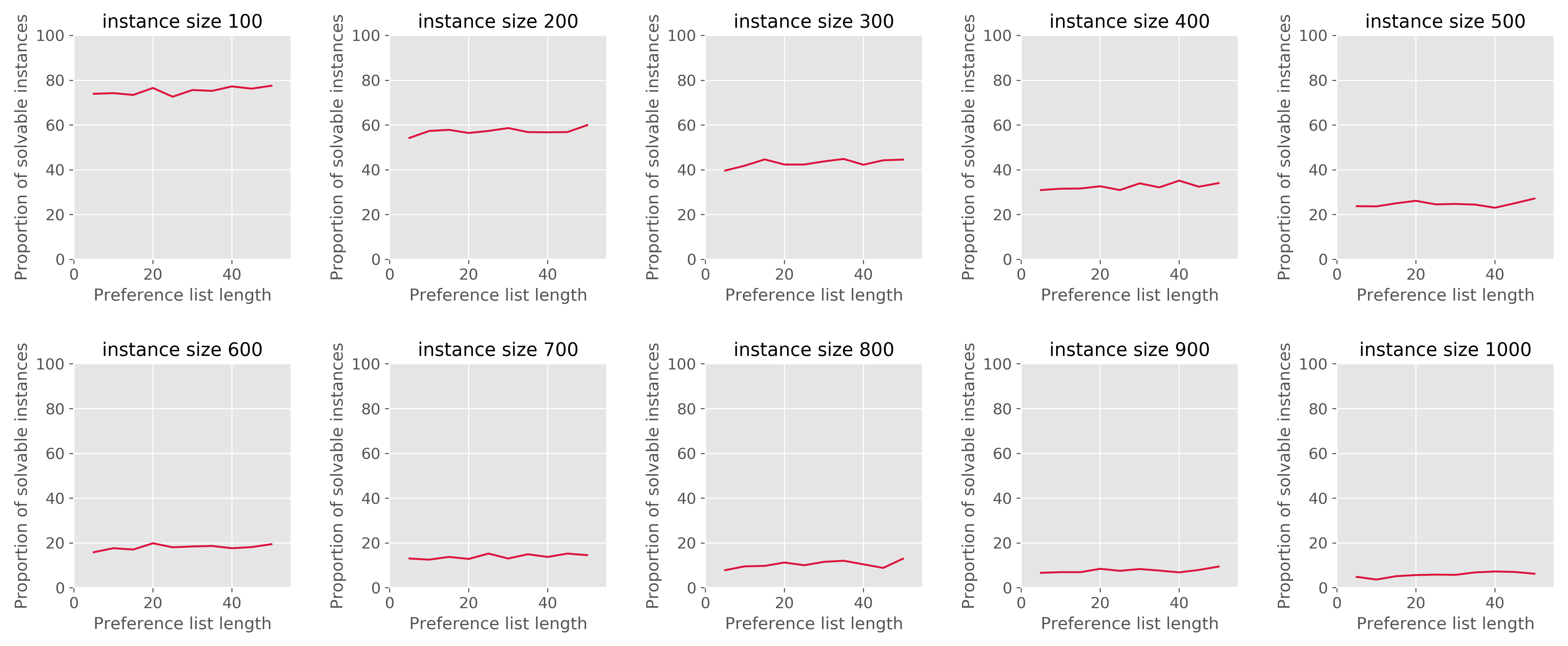} 
\caption{\label{super-experiment1} Proportion of instances that admit a super-stable matching as the size of the instance increases while varying the length of the preference lists with tie density fixed at $0.005$ in both the students' and lecturers' preference lists.}
\end{figure}
\vspace{-0.4in}
\begin{table}[H]
\setlength{\tabcolsep}{0.4em}
\renewcommand*{\arraystretch}{1.2}
\centering
\caption{Time (in seconds) for our algorithm's implementation to terminate averaged over $1000$ for each instance size, with the length of each student's preference list fixed at $50$.}
\label{fig:super-time-table}
\begin{tabular}{c|cccccccccc}
\hline\noalign{\smallskip}
$n_1$ & $100$ & $200$ & $300$ & $400$ & $500$ & $600$ & $700$ & $800$ & $900$ & $1000$ \\ 
\noalign{\smallskip}\hline\noalign{\smallskip}
Time & $0.017$ & $0.046$ & $0.082$ & $0.120$ & $0.160$ & $0.203$ & $0.248$ & $0.298$ & $0.349$ & $0.399$ \\ 
\noalign{\smallskip}\hline
\end{tabular} 
\end{table}

\subsubsection{Experiment 2}
In our second experiment, we investigated how the variation in tie density in both the students' and lecturers' preference lists affects the existence of a super-stable matching. To achieve this, we varied the tie density in the students' preference lists $t_{d_s} \; (0 \leq t_{d_s} \leq 0.05)$ and the tie density in the lecturers' preference lists $t_{d_l} \; (0 \leq t_{d_l} \leq 0.05)$, both in increments of $0.005$. For each pair of tie densities in $t_{d_s} \times t_{d_l}$, we randomly-generated $1000$ {\sc spa-st} instances for various values of $n_1 \; (100 \leq n_1 \leq 1000)$ in increments of $100$. For each of these instances, we maintained the same ratio of projects, lecturers, project capacities and lecturer capacities as in Experiment 1. Considering our discussion from Experiment 1, we fixed the length of each student's preference list at $50$. 

The result displayed in Fig.~\ref{super-experiment2} shows that increasing the tie density in both the students' and lecturers' preference lists reduces the proportion of instances that admit a super-stable matching. In fact, this proportion reduces further as the size of the instance increases. When ties occur only in the lecturers' preference lists, we found that a significantly higher proportion of instances admit a super-stable matching -- about $74\%$ of the randomly-generated {\sc spa-st} instances involving $1000$ students admitted a super-stable matching. The confidence interval for this value is $(0.71, 0.77)$. However, the reverse is the case when ties occur only in the students' preference lists. We have no explanation for this outcome.

\begin{figure}[H]
\centering
\includegraphics[scale=0.085]{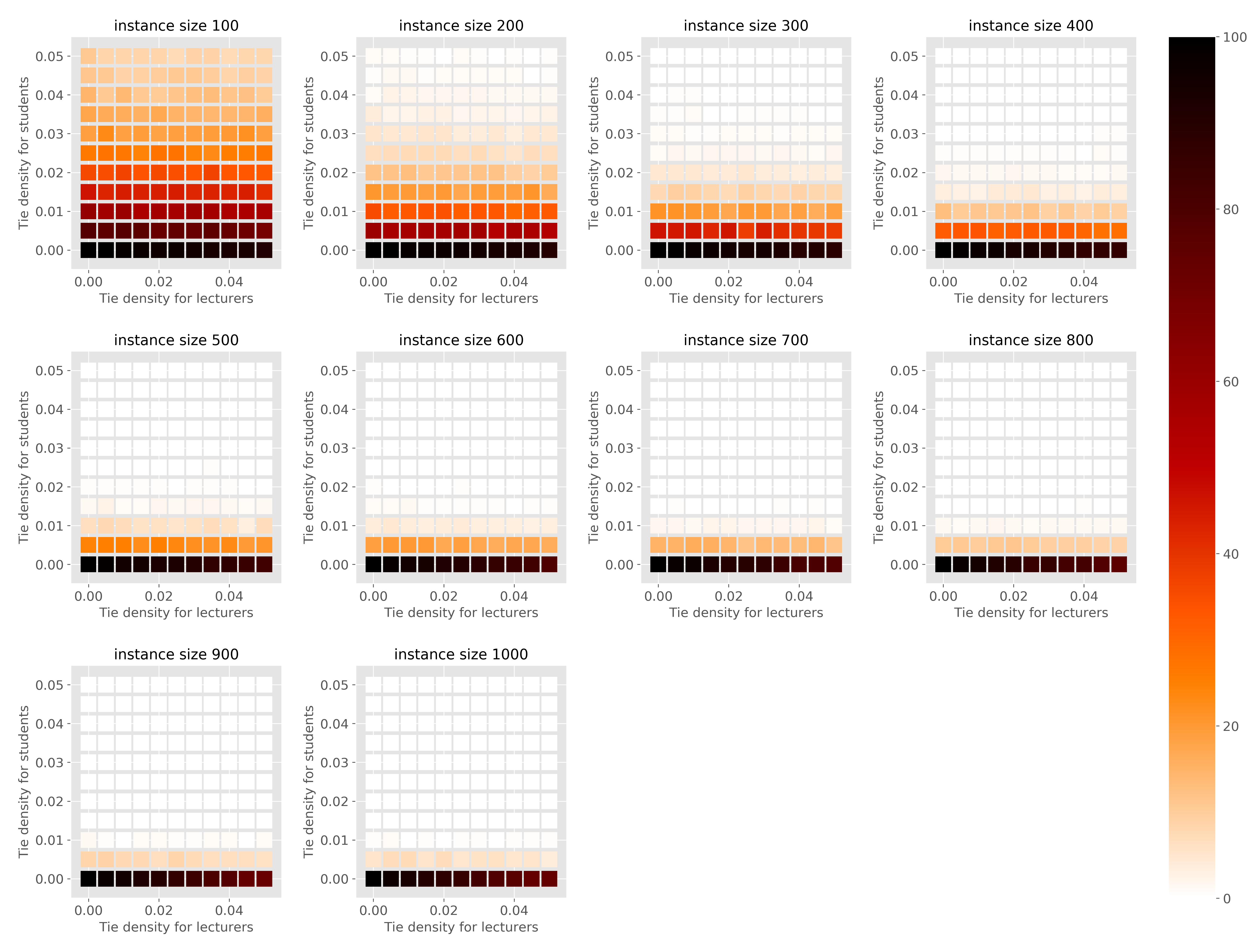} 
\caption{\label{super-experiment2} Result for Experiment 2. Each of the coloured square boxes represents the proportion of the $1000$ randomly-generated {\sc spa-st} instances that admit a super-stable matching, with respect to the tie density in the students' and lecturers' preference lists. See the colour bar transition, as this proportion ranges from dark ($100\%$) to light ($0\%$).}
\end{figure}

%

%
\section{Discussions and Concluding Remarks}
\label{section:conclusions}
In this paper, we have described a linear-time algorithm to find a super-stable matching or report that no such matching exists, given an instance of {\sc spa-st}. We established that for instances that do admit a super-stable matching, our algorithm produces the student-optimal super-stable matching, in the sense that each assigned student has the best project that she could obtain in any super-stable matching. We leave open the formulation of a lecturer-oriented counterpart to our algorithm.

Further, we carried out an empirical evaluation of our algorithm's implementation. The purpose of our experiments was to investigate how the nature of the preference lists affects the existence (or otherwise) of super-stable matchings in an arbitrary instance of {\sc spa-st}. Based on the instances we generated randomly, the experimental results suggest that as we increase the size of the instance and the density of ties in the preference lists, the likelihood of a super-stable matching existing decreases. There was no significant uplift in this likelihood even as we increased the length of the students' preference lists. When the ties occur only in the lecturers' preference lists, we found that a significantly higher proportion of instances admit a super-stable matching. However, the reverse is the case when the ties occur only in the students' preference lists.

Given that there are typically more students than lecturers in practical applications, it could be that only lecturers are permitted to have some form of indifference over the students that they find acceptable, whilst each student might be able to provide a strict ordering over what may be a small number of projects that she finds acceptable. Further evaluation of our algorithm could investigate how other parameters (e.g., the popularity of some projects, or the position of the ties in the preference lists) affect the existence of a super-stable matching. It would also be interesting to examine the existence of super-stable matchings in real {\sc spa-st} datasets.

From a theoretical perspective, the following are other directions for future work. Let $I$ be an arbitrary instance of {\sc spa-st}. 
\begin{enumerate}
\item Can we formalise the results on the probability of a super-stable matching existing in $I$? As mentioned in Section \ref{introduction}, this question has been partially explored for the Stable Roommates problem \cite{PI94}.

\item Is there a characterisation of the set of super-stable matchings in $I$ in terms of a lattice structure? It is known that the set of super-stable matchings in an instance of {\sc smt} forms a distributive lattice under the dominance relation \cite{Man02,Spi95}. To generalise this structural result for {\sc spa-st}, ideas from \cite{Man02,Spi95} would certainly be useful. 
\end{enumerate}

\section*{Acknowledgements}
The authors would like to thank Frances Cooper, Kitty Meeks, Patrick Prosser, and also the anonymous reviewers, for valuable comments that helped to improve the presentation of this paper.





\newpage
\begin{subappendices}

\renewcommand{\thesection}{\Alph{section}}%

\section{An IP model for super-stability in {\small SPA-ST}}
\label{appendixA}
\subsection{Introduction}
In this section, we describe an IP model for super-stability in {\sc spa-st}. Although a super-stable matching in an instance of {\sc spa-st} can be found in polynomial-time (as illustrated by Theorem \ref{thrm:super-optimality}), our reason for this is purely experimental. Let $I$ be an instance of {\sc spa-st} involving a set $\mathcal{S} = \{s_1, s_2, \ldots, s_{n_1}\}$ of students, a set $\mathcal{P} = \{p_1, p_2, \ldots, p_{n_2}\}$ of projects and a set $\mathcal{L} = \{l_1, l_2, \ldots, l_{n_3}\}$ of lecturers. We construct an IP model $J$ of $I$ as follows. Firstly, we create binary variables $x_{i, j} \in \{0, 1\}$ $(1 \leq i \leq n_1, 1 \leq j \leq n_2)$ for each acceptable pair $(s_i, p_j) \in \mathcal{S} \times \mathcal{P}$ such that $x_{i, j}$ indicates whether $s_i$ is assigned to $p_j$ in a solution or not. Henceforth, we denote by $S$ a solution in the IP model $J$, and we denote by $M$ the matching derived from $S$ in the following natural way: if $x_{i,j} = 1$ under $S$ then $s_i$ is assigned to $p_j$ in $M$, otherwise $s_i$ is not assigned to $p_j$ in $M$.

\subsection{Constraints}
In this section, we give the set of constraints to ensure that the assignment obtained from a feasible solution in $J$ is a matching, and that the matching admits no blocking pair.
\paragraph{\textbf{Matching constraints.}}
The feasibility of a matching can be ensured with the following three set of constraints.
\begin{align}
\label{ineq:spa-st-ip-studentassignment}
\sum\limits_{p_{j} \in A_{i}} x_{i,j}  \leq 1 &\qquad (1 \leq i \leq n_1), \\
\label{ineq:spa-st-ip-projectcapacity}
\sum\limits_{i = 1}^{n_1} x_{i,j}   \leq c_j   & \qquad (1 \leq j \leq n_2), \\
\label{ineq:spa-st-ip-lecturercapacity}
\sum\limits_{i = 1}^{n_1} \; \sum\limits_{p_{j} \in P_k} x_{i,j}  \leq d_k  & \qquad (1 \leq k \leq  n_3)\enspace.
\end{align}

Note that Inequality \eqref{ineq:spa-st-ip-studentassignment} ensures that each student $s_i \in \mathcal{S}$ is not assigned to more than one project, while Inequalities \eqref{ineq:spa-st-ip-projectcapacity} and \eqref{ineq:spa-st-ip-lecturercapacity} ensure that the capacity of each project $p_j \in \mathcal{P}$ and each lecturer $l_k \in \mathcal{L}$ is not exceeded.


Given an acceptable pair $(s_i, p_j)$, we define $\rank(s_i, p_j)$, the \textit{rank} of $p_j$ on $s_i$'s preference list, to be $r+1$, where $r$ is the number of projects that $s_i$ prefers to $p_j$. \label{rank} Clearly, projects that are tied together on $s_i$'s preference list have the same rank. Given a lecturer $l_k \in \mathcal{L}$ and a student $s_i \in \mathcal{L}_k$, we define $\rank(l_k, s_i)$, the \textit{rank} of $s_i$ on $l_k$'s preference list, to be $r+1$, where $r$ is the number of students that $l_k$ prefers to $s_i$. Similarly, students that are tied together on $l_k$'s preference list have the same rank. With respect to an acceptable pair $(s_i, p_j)$, we define $S_{i,j} = \{p_{j'} \in A_i: \rank(s_i, p_{j'}) < \rank(s_i, p_j)\}$, the set of projects that $s_i$ prefers to $p_j$. Let $l_k$ be the lecturer who offers $p_j$. We also define $T_{i,j,k} = \{s_{i'} \in \mathcal{L}_{k}^{j}: \rank(l_k, s_{i'}) < \rank(l_k,  s_{i})\}$, the set of students that are better than $s_i$ on the projected preference list of $l_k$ for $p_j$. Finally, we define $D_{i,k} = \{s_{i'} \in \mathcal{L}_{k}: \rank(l_k,  s_{i'}) < \rank(l_k, s_{i})\}$, the set of students that are better than $s_i$ on $l_k$'s preference list. 

In what follows, we fix an arbitrary acceptable pair $(s_i, p_j)$ and we enforce constraints to ensure that $(s_i, p_j)$ does not form a blocking pair for the matching $M$. Henceforth, $l_k$ is the lecturer who offers $p_j$. 

\paragraph{\textbf{Blocking pair constraints.}} First, we define $\theta_{i,j} = 1 - x_{i,j} - \sum\limits_{p_{j'} \in S_{i, j}}x_{i,j'}$. Intuitively, $\theta_{i,j} = 1$ if and only if $s_i$ is unassigned in $M$, or  $s_i$ prefers $p_j$ to $M(s_i)$ or is indifferent between them. Henceforth, if $(s_i, p_j)$ forms a blocking pair for $M$ then we refer to $(s_i, p_j)$ as a blocking pair of type (i), type (ii) or type (iii), according as $(s_i, p_j)$ satisfies condition (i), (ii), or (iii) of Definition \ref{definition:super-stability}, respectively. We describe the constraints to avoid these types of blocking pair as follows.
\paragraph{\textbf{Type (i)}. \label{type-i}}
First, we create a binary variable $\alpha_j$ in $J$ such that if $p_j$ is undersubscribed in $M$ then $\alpha_j = 1$. We enforce this condition by imposing the following constraint.
\begin{eqnarray}
\label{ineq:spa-st-ip-project-under}
c_j \alpha_j \geq c_j - \sum\limits_{i' = 1}^{n_1} x_{i',j},
\end{eqnarray}
where $\sum_{i' = 1}^{n_1} x_{i',j} = |M(p_j)|$. If $p_j$ is undersubscribed in $M$ then the RHS of Inequality \eqref{ineq:spa-st-ip-project-under} is at least $1$ and this implies that $\alpha_j = 1$, otherwise $\alpha_j$ is not constrained. Next, we create a binary variable $\beta_k$ in $J$ such that if $l_k$ is undersubscribed in $M$ then $\beta_k = 1$. We enforce this condition by imposing the following constraint:
\begin{eqnarray}
\label{ineq:spa-st-ip-lecturerunder}
d_k\beta_k \geq d_k - \sum\limits_{i' = 1}^{n_1} \; \sum\limits_{p_{j'} \in P_k} x_{i',j'},
\end{eqnarray}
where $\sum\limits_{i' = 1}^{n_1} \; \sum\limits_{p_{j'} \in P_k} x_{i',j'} = |M(l_k)|$. If $l_k$ is undersubscribed in $M$ then the RHS of Inequality \eqref{ineq:spa-st-ip-lecturerunder} is at least $1$ and this implies that $\beta_k = 1$, otherwise $\beta_k$ is not constrained. The following constraint ensures that $(s_i, p_j)$ does not form a type (i) blocking pair for $M$.
\begin{align}
\label{ineq:super-bp-type-i}
\Aboxed{
\theta_{i,j} + \alpha_{j} + \beta_k \leq 2\enspace.}
\end{align}
\paragraph{\textbf{Type (ii)}. \label{type-ii}}
We create a binary variable $\eta_{k}$ in $J$ such that if $l_k$ is full in $M$ then $\eta_{k} = 1$. We enforce this condition by imposing the following constraint.
\begin{eqnarray}
\label{ineq:spa-st-ip-lecturerfull}
d_k\eta_{k} \geq \left(1 + \sum\limits_{i' = 1}^{n_1} \; \sum\limits_{p_{j'} \in P_k} x_{i',j'}\right) - d_k\enspace.
\end{eqnarray}
If $l_k$ is full in $M$ then the RHS of Constraint  \eqref{ineq:spa-st-ip-lecturerfull} is at least $1$ and this implies that $\eta_k = 1$, otherwise $\eta_k$ is not constrained. Next, we create a binary variable $\delta_{i,k}$ in $J$ such that if $s_i \in M(l_k)$, or $l_k$ prefers $s_i$ to a worst student in $M(l_k)$ or is indifferent between them, then $\delta_{i,k} = 1$. We enforce this condition by imposing the following constraint.
\begin{eqnarray}
\label{ineq:spa-st-ip-lecturerfull-student}
d_k\delta_{i,k} \geq \sum\limits_{i' = 1}^{n_1} \; \sum\limits_{p_{j'} \in P_k} x_{i',j'} -  \sum\limits_{s_{i'} \in D_{i,k}} \; \sum\limits_{p_{j'} \in P_k}x_{i',j'}\enspace.
\end{eqnarray}
Note that if $s_i \in M(l_k)$ or $l_k$ prefers $s_i$ to a worst student in $M(l_k)$ or $l_k$ is indifferent between them, then the RHS of Constraint  \eqref{ineq:spa-st-ip-lecturerfull-student} is at least 1 and this implies that $\delta_{i,k} = 1$, otherwise $\delta_{i,k}$ is not constrained. The following constraint ensures that $(s_i, p_j)$ does not form a type (ii) blocking pair for $M$.
\begin{align}
\label{ineq:super-bp-type-ii}
\Aboxed{
\theta_{i,j} + \alpha_{j} + \eta_{k} + \delta_{i,k} \leq 3\enspace.}
\end{align}
\paragraph{\textbf{Type (iii)}. \label{type-iii}}
Next we create a binary variable $\gamma_{j}$ in $J$ such that if $p_j$ is full in $M$ then $\gamma_{j} = 1$. We enforce this condition by imposing the following constraint.
\begin{eqnarray}
\label{ineq:spa-st-ip-projectfull}
c_j\gamma_{j} \geq \left( 1 + \sum\limits_{i' = 1}^{n_1} \; x_{i',j} \right) - c_j\enspace.
\end{eqnarray}
where $\sum_{i' = 1}^{n_1} x_{i',j} = |M(p_j)|$. If $p_j$ is full in $M$ then the RHS of Inequality \eqref{ineq:spa-st-ip-projectfull} is at least $1$ and this implies that $\gamma_j = 1$, otherwise $\gamma_j$ is not constrained. Next, we create a binary variable $\lambda_{i,j,k}$ in $J$ such that if $l_k$ prefers $s_i$ to a worst student in $M(p_j)$ or is indifferent between them, then $\lambda_{i,j,k}=1$. We enforce this condition by imposing the following constraint.
\begin{eqnarray}
\label{ineq:spa-st-ip-projectfull-student}
c_j\lambda_{i,j,k} \geq \sum\limits_{i' = 1}^{n_1} x_{i',j} - \sum\limits_{s_{i'} \in T_{i,j,k}} x_{i',j}\enspace.
\end{eqnarray}
Note that if $l_k$ prefers $s_i$ to a worst student in $M(p_j)$ or is indifferent between them, then the RHS of Inequality  \eqref{ineq:spa-st-ip-projectfull-student} is at least 1 and this implies that $\lambda_{i,j,k} = 1$, otherwise $\lambda_{i,j,k}$ is not constrained. The following constraint ensures that $(s_i, p_j)$ does not form a type (iii) blocking pair for $M$.
\begin{align}
\label{ineq:super-bp-type-iii}
\Aboxed{
\theta_{i,j} +  \gamma_j + \lambda_{i,j,k} \leq 2\enspace.}
\end{align}

\subsection{Variables} 
\label{sect:spa-st-ip-variables}
We define a collective notation for each set of variables involved in $J$ as follows:
\begin{center}
\begin{tabular}{p{5cm}p{0.2cm}p{6cm}}
$A = \{ \alpha_{j}: 1 \leq j \leq n_2\}$, &  & $\Gamma = \{ \gamma_{j}: 1 \leq j \leq n_2\}$, \\
$B = \{\beta_{k}: 1 \leq k \leq n_3\}$, &  & $\Delta = \{ \delta_{i,k}: 1 \leq i \leq n_1, 1 \leq k \leq n_3\}$, \\ 
$N = \{\eta_{k}: 1 \leq k \leq n_3\}$, &  & $X = \{ x_{i,j}: 1 \leq i \leq n_1, 1 \leq j \leq n_2\}$, \\ 
\multicolumn{3}{p{12cm}}{$\Lambda = \{\lambda_{i,j,k}: 1 \leq i \leq n_1, 1 \leq j \leq n_2, 1 \leq k \leq n_3 \}$\enspace.}  \\

\end{tabular} 
\end{center}

\subsection{Objective function}
On one hand, all super-stable matchings are of the same size, and thus nullifies the need for an objective function. On the other hand, optimization solvers require an objective function in addition to the variables and constraints in order to produce a solution. The objective function given below involves maximising the summation of all the $x_{i,j}$ binary variables. 
\begin{align}
\label{ineq:super-objectivefunction}
\Aboxed{\max \sum\limits_{i = 1}^{n_1} \; \sum\limits_{p_j \in A_i}x_{i,j}\enspace.}
\end{align}
Finally, we have constructed an IP model $J$ of $I$ comprising the set of integer-valued variables $A, B, N, X, \Gamma, \Delta, \mbox{ and } \Lambda$, the set of Inequalities \eqref{ineq:spa-st-ip-studentassignment} - \eqref{ineq:super-bp-type-iii} and an objective function \eqref{ineq:super-objectivefunction}. Note that $J$ can then be used to construct a super-stable matching in $I$, should one exist. 

\subsection{Correctness of the IP model}
Given an instance $I$ of {\sc spa-st} formulated as an IP model $J$ using the above transformation, we present the following lemmas regarding the correctness of $J$.
\begin{lemma}
\label{lemma:super-solution-stability}
A feasible solution $S$ to $J$ corresponds to a super-stable matching $M$ in $I$.
\end{lemma}
\begin{proof}
Assume firstly that $J$ has a feasible solution $S$. Let $M = \{(s_i, p_j) \in \mathcal{S} \times \mathcal{P}: x_{i,j} = 1\}$ be the assignment in $I$ generated from $S$. We note that Inequality \eqref{ineq:spa-st-ip-studentassignment} ensures that each student is assigned in $M$ to at most one project. Moreover, Inequalities \eqref{ineq:spa-st-ip-projectcapacity} and \eqref{ineq:spa-st-ip-lecturercapacity} ensures that the capacity of each project and lecturer is not exceeded in $M$. Thus $M$ is a matching. We will prove that Inequalities \eqref{ineq:spa-st-ip-project-under} - \eqref{ineq:super-bp-type-iii} ensures that $M$ admits no blocking pair.

Suppose for a contradiction that there exists some acceptable pair $(s_i, p_j)$ that forms a blocking pair for $M$, where $l_k$ is the lecturer who offers $p_j$. This implies that either $s_i$ is unassigned in $M$ or $s_i$ prefers $p_j$ to $M(s_i)$ or is indifferent between them. Thus $\sum_{p_{j'} \in S_{i,j}} x_{{i},{j'}} = 0$. Moreover, since $s_i$ is not assigned to $p_j$ in $M$, we have that $x_{i,j} = 0$. Thus $\theta_{i,j} = 1$.

Now suppose $(s_i, p_j)$ forms a type (i) blocking pair for $M$. Then each of $p_j$ and $l_k$ is undersubscribed in $M$. Thus $\sum_{i' = 1}^{n_1} x_{i',j} < c_j$ and $\sum_{i' = 1}^{n_1} \; \sum_{p_{j'} \in P_k} x_{i',j'}$ $< d_k$. This implies that the RHS of Inequality \eqref{ineq:spa-st-ip-project-under} and the RHS of Inequality \eqref{ineq:spa-st-ip-lecturerunder} is strictly greater than $0$. Moreover, since $S$ is a feasible solution to $J$, $\alpha_j = \beta_k = 1$. Hence, the LHS of Inequality \eqref{ineq:super-bp-type-i} is strictly greater than $2$, a contradiction to the feasibility of $S$.

Now suppose $(s_i, p_j)$ forms a type (ii) blocking pair for $M$. Then $p_j$ is undersubscribed in $M$ and as explained above, $\alpha_j = 1$. Also, $l_k$ is full in $M$ and this implies that the RHS of Inequality \eqref{ineq:spa-st-ip-lecturerfull} is strictly greater than $0$. Since $S$ is a feasible solution, we have that $\eta_k = 1$. Furthermore, either $s_i \in M(l_k)$ or $l_k$ prefers $s_i$ to a worst student in $M(l_k)$ or $l_k$ is indifferent between them. In any of these cases, the RHS of Inequality \eqref{ineq:spa-st-ip-lecturerfull-student} is strictly greater than $0$. Thus $\delta_{i,k} = 1$, since $S$ is a feasible solution. Hence the LHS of Inequality \eqref{ineq:super-bp-type-ii} is strictly greater than 3, a contradiction to the feasibility of $S$.

Finally, suppose $(s_i, p_j)$ forms a type (iii) blocking pair for $M$. Then $p_j$ is full in $M$ and thus the RHS of Inequality \eqref{ineq:spa-st-ip-projectfull} is strictly greater than $0$. Since $S$ is a feasible solution, we have that $\gamma_j = 1$. In addition, $l_k$ prefers $s_i$ to a worst student in $M(p_j)$ or is indifferent between them. This implies that the RHS of Inequality \eqref{ineq:spa-st-ip-projectfull-student} is strictly greater than $0$. Thus $\lambda_{i,j,k} = 1$, since $S$ is a feasible solution. Hence the LHS of Inequality \eqref{ineq:super-bp-type-iii} is strictly greater than 2, a contradiction to the feasibility of $S$. Hence $M$ admits no blocking pair; and hence, $M$ is a super-stable matching in $I$.
\qed \end{proof}

\begin{lemma}
\label{lemma:super-stability-solution}
A super-stable matching $M$ in $I$ corresponds to a feasible solution $S$ to $J$.
\end{lemma}
\begin{proof}
Let $M$ be a super-stable matching in $I$. First we set all the binary variables involved in $J$ to $0$. For each $(s_i, p_j) \in M$, we set $x_{i,j} = 1$. Since $M$ is a matching, it is clear that Inequalities \eqref{ineq:spa-st-ip-studentassignment} - \eqref{ineq:spa-st-ip-lecturercapacity} is satisfied. For any acceptable pair $(s_i, p_j) \in (\mathcal{S} \times \mathcal{P}) \setminus M$ such that $s_i$ is unassigned in $M$ or $s_i$ prefers $p_j$ to $M(s_i)$ or is indifferent between them, we set $\theta_{i,j} = 1$. For any project $p_j \in \mathcal{P}$ such that $p_j$ is undersubscibed in $M$, we set $\alpha_j = 1$ and thus Inequality \eqref{ineq:spa-st-ip-project-under} is satisfied. For any lecturer $l_k \in \mathcal{L}$ such that $l_k$ is undersubscribed in $M$, we set $\beta_k = 1$ and thus Inequality \eqref{ineq:spa-st-ip-lecturerunder} is satisfied. 

Now, for Inequality \eqref{ineq:super-bp-type-i} not to be satisfied, its LHS must be strictly greater than 2. This would only happen if there exists some $(s_i, p_j) \in (\mathcal{S} \times \mathcal{P}) \setminus M$, where $l_k$ is the lecturer who offers $p_j$, such that $\theta_{i,j} = 1$, $\alpha_j = 1$ and $\beta_k = 1$. This implies that either $s_i$ is unassigned in $M$ or $s_i$ prefers $p_j$ to $M(s_i)$ or is indifferent between them, and each of $p_j$ and $l_k$ is undersubscribed in $M$. Thus $(s_i, p_j)$ forms a type (i) blocking pair for $M$, a contradiction to the super-stability of $M$. Hence, Inequality \eqref{ineq:super-bp-type-i} is satisfied.

For any lecturer $l_k \in \mathcal{L}$ such that $l_k$ is full in $M$, we set $\eta_k = 1$. Thus Inequality \eqref{ineq:spa-st-ip-lecturerfull} is satisfied. Let $(s_i, p_j)$ be an acceptable pair such that $p_j \in P_k$ and $(s_i, p_j) \notin M$. If $s_i \in M(l_k)$ or $l_k$ prefers $s_i$ to a worst student in $M(l_k)$ or is indifferent between them, we set $\delta_{i,k} = 1$. Thus Inequality \eqref{ineq:spa-st-ip-lecturerfull-student} is satisfied. Suppose Inequality \eqref{ineq:super-bp-type-ii} is not satisfied. Then there exists $(s_i, p_j) \in (\mathcal{S} \times \mathcal{P}) \setminus M$, where $l_k$ is the lecturer who offers $p_j$, such that $\theta_{i,j} = 1$, $\alpha_j = 1$, $\eta_k = 1$ and $\delta_{i,k} = 1$. This implies that either $s_i$ is unassigned in $M$ or $s_i$ prefers $p_j$ to $M(s_i)$ or is indifferent between them. In addition, $p_j$ is undersubscribed in $M$, $l_k$ is full in $M$ and either $s_i \in M(l_k)$ or $l_k$ prefers $s_i$ to a worst student in $M(l_k)$ or is indifferent between them. Thus $(s_i, p_j)$ forms a type (ii) blocking pair for $M$, a contradiction to the super-stability of $M$. Hence Inequality \eqref{ineq:super-bp-type-ii} is satisfied.

Finally, for any project $p_j \in \mathcal{P}$ such that $p_j$ is full in $M$, we set $\gamma_j = 1$. Thus Inequality \eqref{ineq:spa-st-ip-projectfull} is satisfied. Let $l_k$ be the lecturer who offers $p_j$ and let $(s_i, p_j)$ be an acceptable pair. If $l_k$ prefers $s_i$ to a worst student in $M(p_j)$ or is indifferent between them, we set $\lambda_{i,j,k} = 1$. Thus Inequality \eqref{ineq:spa-st-ip-projectfull-student} is satisfied. Suppose Inequality \eqref{ineq:super-bp-type-iii} is not satisfied. Then there exists some $(s_i, p_j) \in (\mathcal{S} \times \mathcal{P}) \setminus M$ such that $\theta_{i,j} = 1$, $\gamma_j = 1$ and $\lambda_{i,j,k} = 1$. This implies that either $s_i$ is unassigned in $M$ or $s_i$ prefers $p_j$ to $M(s_i)$ or is indifferent between them. In addition, $p_j$ is full in $M$ and $l_k$ prefers $s_i$ to a worst student in $M(p_j)$ or is indifferent between them. Thus $(s_i, p_j)$ forms a type (iii) blocking pair for $M$, a contradiction to the super-stability of $M$. Hence, Inequality \eqref{ineq:super-bp-type-iii} is satisfied.
Hence $S$, comprising the above assignments of values to the variables in $A \cup B \cup N \cup X \cup \Gamma \cup \Delta \cup \Lambda$, is a feasible solution to $J$.
\qed \end{proof}

The following theorem is a consequence of Lemmas \ref{lemma:super-solution-stability} and \ref{lemma:super-stability-solution}.
\begin{theorem}
\label{theorem:super-stable-solution}
Let $I$ be an instance of {\sc spa-st} and let $J$ be the IP model for $I$ as described above. A feasible solution to $J$ corresponds to a super-stable matching in $I$. Conversely, a super-stable matching in $I$ corresponds to a feasible solution to $J$.
\end{theorem}

\end{subappendices}

\end{document}